\documentclass[a4paper,UKenglish]{llncs}
 \usepackage{proof}
\usepackage{amssymb}
\usepackage{amsmath}
\usepackage{stmaryrd}
\usepackage[all]{xy}
\DeclareMathAlphabet{\mathsl}{OT1}{cmr}{m}{sl}

\usepackage{color}
\usepackage{url}
\usepackage{ifpdf} 
\usepackage{lineno}
\usepackage{balance}
\usepackage{graphicx}
\usepackage{lipsum}
\usepackage{wrapfig}
\usepackage{microtype}%if unwanted, comment out or use option "draft"
\usepackage[dvipsnames]{xcolor}
\usepackage{multirow}
\usepackage{booktabs}
\usepackage[font={small}]{caption, subfig}
\usepackage{times}
\newcommand{\red}[1]{\textcolor{Mahogany}{#1}}
\newcommand{\blue}[1]{\textcolor{MidnightBlue}{#1}}
\long\def\comment#1{}
\newcommand{\eg}{{\em e.g.}}
\newcommand{\ie}{{\em i.e.}}

\newcommand\uSize{\ensuremath{k}}
\newcommand\iSize{\ensuremath{m}}

\newcommand{\Ascr}{\mathcal{A}}
\newcommand{\Pscr}{\mathcal{P}}
\newcommand{\Uscr}{\mathcal{U}}

\newcommand{\Cscr}{\mathcal{C}}

\newcommand{\Sscr}{\mathcal{S}}
\newcommand{\Rscr}{\mathcal{R}}
\newcommand{\Tscr}{\mathcal{T}}
\newcommand{\Wscr}{\mathcal{W}}

\newcommand{\Zscr}{\mathcal{Z}}
\newcommand{\Gscr}{\mathcal{G}}
\newcommand{\Fscr}{\mathcal{F}}

\newcommand{\CS}{\mathcal{CS}}

\newcommand{\GS}{\mathcal{GS}}

\newcommand{\bluend}{ {\color{blue} $\blacktriangleleft$ }}

\newcommand{\arr}{\rotatebox[origin=c]{-90}{$\leftrightsquigarrow$}}

\newcommand\next{\mathcal{N}}
\newcommand\critical{\mathcal{X}}

\newcommand\goal{\mathcal{G}}

\newcommand{\tup}[1]{\langle#1\rangle}

\newcommand\lra{\longrightarrow}

\newcommand\dec[1]{dec(#1)}
\newcommand\Dmax{D_{max}}
\usepackage{tweaklist}
\renewcommand{\enumhook}{
  \setlength{\topsep}{0em}
  \setlength{\itemsep}{0em}
  \setlength{\partopsep}{0em}
  \setlength{\parskip}{0em}
  \setlength{\parsep}{0em}
  \addtolength{\leftmargin}{-0.75em}
}
\renewcommand{\itemhook}{
  \setlength{\topsep}{0em}
  \setlength{\itemsep}{0em}
  \setlength{\partopsep}{0em}
  \setlength{\parskip}{0em}
  \setlength{\parsep}{0em}
  \addtolength{\leftmargin}{-0.75em}
}
\title{Compliance in Real Time Multiset Rewriting Models}

\author{Max Kanovich\inst{1,6} Tajana Ban Kirigin\inst{2} Vivek Nigam\inst{3,4} Andre Scedrov\inst{5,6} and Carolyn Talcott\inst{7}}
\institute{
University College, London, UK,
\email{m.kanovich@ucl.ac.uk}
\and
University of Rijeka, Department of Mathematics, HR,
\email{bank@math.uniri.hr}
\and
Federal University of Para\'iba, Jo\~ao Pessoa, Brazil,
\email{vivek@ci.ufpb.br}
\and
fortiss, Germany, \email{nigam@fortiss.org}
\and
University of Pennsylvania, Philadelphia,
USA,
\email{scedrov@math.upenn.edu}
 \and National Research University Higher School of Economics, Moscow, Russia
\and
SRI International, USA, \email{clt@csl.sri.com} 
}

\begin{document}
%\begin{center}
    %  \Large\textbf{{Real Time Multiset Rewriting Model in the Verification of Time-Sensitive Distributed Systems}}
%   \end{center}

\maketitle

\begin{abstract}
The notion of compliance in Multiset Rewriting Models (MSR) has been introduced for untimed models and for models with discrete time.
In this paper we revisit the notion of compliance and adapt it to fit with additional nondeterminism specific for dense time domains.  Existing MSR with dense time are extended with critical configurations and non-critical traces, that is, traces involving no critical configurations.
Complexity of related {\em non-critical reachability problem} is investigated.
 Although this problem is undecidable in general, we prove that for balanced MSR with dense time the non-critical reachability problem is PSPACE-complete.
\end{abstract}

% \vspace{-2mm}
\section{%Timed
Multiset Rewriting Systems with Real Time}
\label{sec:timedmsr}
% \vspace{-2mm}

We  follow~\cite{kanovich15post} in formalizing dense time in the multiset rewriting framework.

\vspace{0.5em}
Assume a finite first-order typed alphabet, $\Sigma$, with variables, constants, function and predicate symbols.
Terms and %facts
formulas are constructed as usual (see~\cite{enderton}) by applying symbols of correct type (or sort).

 %For instance,  
If $P$ is a predicate of type $\tau_1 \times \tau_2 \times \cdots \times \tau_n \rightarrow o$, where $o$ is the type for propositions, and $u_1, \ldots, u_n$ are terms of types $\tau_1, \ldots, \tau_n$, respectively, then $P(u_1, \ldots, u_n)$ is a \emph{fact}.

 A fact is grounded if it does not contain any variables. 
We assume that the alphabet contains the constant $z : Nat$ denoting zero and the function $s : Nat \to Nat$ denoting the successor function. Whenever it is clear from the context, we write $n$ for $s^n(z)$ and $(n + m)$ for $s^n(s^m(z))$. 

Additionally, we allow an unbounded number of fresh values~\cite{cervesato99csfw,durgin04jcs} to be involved.

\vspace{0.5em}
In order to specify timed systems, to each fact we attach a timestamp denoting time.
\emph{Timestamped facts} are of the form $F@t$, where $F$ is a fact and $t \in \mathbb{R}$ is a \emph{non-negative real number}  called {\em timestamp}.\footnote{Notice that timestamps are \emph{not} constructed by using the successor function or any other function from the alphabet.} 
Similarly, time variables denoting timestamps,  such as variable $T$ in $F@T$, range over non-negative real numbers.

For simplicity, instead of timestamped facts, we often simply say facts. Also,  when we want to emphasize a difference between a fact $F$, and a timestamped fact $F@t$, we say that $F$ is an \emph{untimed fact}.

There is a special predicate symbol $Time$ with arity zero, which will be used to represent global time.
For example, the  fact $Time@10.4$ denotes that the current global  time of the system is $10.4$.

Given $Time@t$, we say that a fact $F@t_F$ is a \emph{future fact} when its timestamp is greater than the global time $t$, \ie~when $t_F>t$.
Similarly,  a fact $F@t_F$ is a \emph{past fact} when $t_F<t$, and  a fact $F@t_F$ is a \emph{present fact} when $t_F=t$.

\vspace{0.5em}
 A {\em configuration} is a multiset of ground timestamped facts, 
$$\Sscr = \{~Time@t, ~F_1@t_1, \ldots, ~F_n@t_n~\}$$%,
 with a single occurrence of a $Time$ fact.

\vspace{0.5em}
Configurations are to be interpreted as states of the system.
Configurations are modified by multiset rewrite rules which can be interpreted as actions of the system. 
There is only one rule, $Tick$,  that modifies global time:
\begin{equation}
\label{eq:tick-real}
  Time@T \lra  Time@(T + \varepsilon)
\end{equation}
where $T$ is a time variable and $\varepsilon$ can be instantiated by any non-negative real number. 
We also write $Tick_{\varepsilon}$ when we refer to the  $Tick$ rule (\ref{eq:tick-real}) for a specific $\varepsilon$.
Applied to a configuration, $\{\,Time@t, \,F_1@t_1, \ldots, \,F_n@t_n\, \}$, $Tick_{\varepsilon}$ advances global time by $\varepsilon$, resulting in  configuration $\{\, Time@(t +\varepsilon), \,F_1@t_1, \ldots, \,F_n@t_n\, \}$. 

\vspace{0,5em}
We  point out that the $Tick$ rule changes only the timestamp of the fact $Time$, while the remaining facts in the configuration (those different from $Time$) are unchanged.
%On the other hand,  instantaneous rules do not modify the global time, but may modify the remaining facts of configurations. 

\vspace{0.5em}
The remaining rules are \emph{instantaneous} as they do not modify global time, but may modify the remaining facts of configurations (those different from $Time$). Instantaneous rules have the form:
\begin{equation}
\begin{array}{l}
 Time@T, \,W_1@T_1,\ldots,\,W_p@T_p , \red{\,F_1@T_1'}, \ldots, \red{\,F_n@T_n'} \mid \Cscr 
\lra  \\ \ 
 \exists \vec{X}.\,[ \ 
 Time@T, \,W_1@T_1,\ldots,\,W_p@T_p, \blue{\,Q_1@(T + D_1)}, \ldots, \blue{\,Q_m@(T + D_m)} \, ]
\label{eq:instantaneous-real}
\end{array}
\end{equation}
% \vspace{-2mm} \noindent
where $D_1, \ldots, D_m$ are natural numbers, 
 $\Wscr = \{~W_1@T_1,\ldots,\,W_p@T_p\ \}$ is a multiset of timestamped facts, possibly containing variables,
and  $\Cscr$ is the guard of the rule which  is a set of constraints
involving the time variables  appearing in the rule's pre-condition, \ie~the variables 
$T, T_1, \ldots, T_p, T_1', \ldots, T_n'$.

Constraints may be of the form:
%\vspace{-2mm}
\begin{equation}
\label{eq:constraints}
  T > T' \pm N \quad \textrm{ and } \quad  T = T' \pm N 
\end{equation}
%\vspace{-2.5mm} \noindent
where $T$ and $T'$ are time variables, and $N\in\mathbb{N}$ is a natural number. 

%\footnote
{Here, and in the rest of the paper,  the symbol $\pm$ stands for either $+$ or $-$, that is,  constraints may involve addition or subtraction.}

We use $T' \geq T' \pm N$ to denote the disjunction of $T > T' \pm N$ and $T = T' \pm N$.  
All time variables in the guard of a rule are assumed to appear in the rule's pre-condition.  

\vspace{0,5em}
Finally, the variables $\vec{X}$ that are existentially quantified in the rule (Equation~\ref{eq:instantaneous-real})
are to be replaced by fresh values, also called \emph{nonces} in protocol security literature~\cite{cervesato99csfw,durgin04jcs}. 
As in our previous work~\cite{kanovich13ic}, we use nonces whenever a unique identification is required, for example for some protocol session or transaction identification.

\vspace{0,5em}
A rule $\Wscr \mid \Cscr \lra \exists \vec{X}.\Wscr'$ can be applied to a configuration $\Sscr$ if there is a ground substitution $\sigma$, where the variables in $\vec{X}$ are fresh, such that $\Wscr\sigma \subseteq \Sscr$ and $\Cscr\sigma$ is true. The resulting configuration is~$\big((\Sscr \setminus \Wscr) \cup \Wscr'\big)\sigma$.
% An instantaneous rule can only be applied if all the constraints in its guard are satisfied. For example, the following rule
% \begin{equation}
% \begin{small}
%    Time@T, \red{Dr(Id,X,Y,E + 1)@T} \lra Time@T, \blue{Dr(Id,X,Y+1,E)@(T+1)}
% \end{small}
%    \label{eq:rule-north}
% \end{equation}
% specifies that if the drone $Id$ is at the position $(X,Y)$ at the current time and its battery has positive energy $E+1$, it can move a position to the north $(X,Y+1)$ and it should arrive at that position at time $T+1$. Applying this rule to the configuration $\Sscr_1$ given in  (Equation \ref{conf-example-1}) by instantiating $Id$ with $d1$ results in the configuration:
% \[
% \Sscr_2 = \left\{ 
% \begin{array}{l}
%  Time@4, Dr(d1,1,3,9)@5, Dr(d2,5,5,8)@4, P(p1,1,1)@3, P(p2,5,6)@0
% \end{array}
% \right\}
% \]  

More precisely, given some rule $r$, an instance of a rule is obtained by substituting all variables appearing in the pre- and post-condition of the rule with constants. This substitution applies to variables appearing in terms inside facts, variables representing fresh values, as well as time variables used in specifying timestamps of facts. An instance of an instantaneous rule can only be applied if all the constraints in its guard  are satisfied.

In order to express timed properties of the system, besides being attached to the rules, constraints may be attached to configurations. In particular, constraints may be used to express specific timed properties of configurations. For example,
\[
Time@T, {Deadline(p)@T'}, %Dr(d,x_1,y_1,e)@T_1, 
\Wscr \mid \{~ T+7=T'~\}
\]
represents a configuration where a deadline of process $p$ is in 7 time units.
%point of interest has not been photographed for at least 7 time units.

 %such as initial and critical configurations that we will introduce later on.

\vspace{0,5em}
Following \cite{durgin04jcs} we say that a fact is \emph{consumed} by some rule $r$ if that fact occurs more times in $r$ on the left side  than on the right side. A fact is \emph{created} by some rule $r$ if that fact occurs more times in $r$ on the right side than on the left side. 
Hence, ${F_1@T_1'}, \ldots, {F_n@T_n'}$ are consumed by the rule (\ref{eq:instantaneous-real}) and ${Q_1@(T + D_1)}, \ldots, {Q_m@(T + D_m)}$ are created by that rule. In a rule, we usually color \red{red} the consumed facts and \blue{blue} the created facts.

\vspace{2mm}
We write ~$\Sscr \lra_r \Sscr'$\ for the one-step relation where configuration $\Sscr$ is rewritten to $\Sscr'$ using an instance of rule $r$. 
For a set of rules $\Rscr$, we define ~$\Sscr \lra_\Rscr^* \Sscr'$~ as the transitive reflexive closure of the one-step relation on all rules in $\Rscr$. We elide the subscript \ $\Rscr$, when it is clear from the context, and simply write  ~$\Sscr \lra_* \Sscr'$. 

%Notice that while $Tick$ rule modifies only the timestamp of fact $Time$, the remaining rules do not modify the global time and may only modify facts different from $Time$.

\begin{definition}
A timed MSR system with dense time  %is a tuple
 $\Tscr$ is a set of rules containing only instantaneous rules (Eq.~\ref{eq:instantaneous-real}) and the $Tick$ rule (%Equation
Eq.~\ref{eq:tick-real}).
\end{definition}

A trace of a timed MSR %$\Ascr$
 is constructed by a sequence of rules. % from $\Ascr$.
%In this paper we consider both finite and infinite traces. 
A \emph{finite trace} of a timed MSR $\Tscr$  starting from an initial configuration $\Sscr_0$ is a sequence
$$
\Sscr_0 \lra \Sscr_1 \lra \Sscr_2 \lra \cdots \lra \Sscr_n   
$$
%and an \emph{infinite trace} of $\Tscr$   starting from an initial configuration $\Sscr_0$ is a sequence
%% of configurations
%$$
%\Sscr_0 \lra \Sscr_1 \lra \Sscr_2 \lra \cdots \lra \Sscr_n \lra \cdots  
%$$
 where %for all ~$i \geq 0$, 
\ $\Sscr_{i} \lra_{r_i} \Sscr_{i+1}$ \ for some $r_i \in \Tscr$, for all ~$i \in \{0, \dots, n\}$ .
Infinite traces can also be considered, as in \cite{kanovich16formats}, but in this paper only finite traces will be used.

\vspace{0,5em}
Notice that by the nature of multiset rewriting there are various aspects of 
non-determinism in the model. For example, different actions and even different instantiations of the 
same rule may be applicable to the same configuration $\Sscr$, which may 
lead to different resulting configurations $\Sscr'$. 

\vspace{0,5em}
There is the additional non-determinism in the dense time model with respect to the discrete time model used in ~\cite{kanovich16formats}, provided by the choice of $\varepsilon$, representing the non-negative real value of time increase.
While in the discrete time model, 
time is advancing using the rule  \
\begin{equation}
\label{eq:tick-1}
Time@T \lra  Time@(T + 1) ,
\end{equation}
where
time always advances by one time unit,
in the dense time model, using  the rule~(Eq.~\ref{eq:tick-real}),
 % provided with the possibility of time advancing 
time can advance by any non-negative real value $\varepsilon$.

\vspace{0,5em}
\begin{remark}
\label{remark-consecutive-ticks}
Notice that the consecutive time advancements $Tick_{\varepsilon_1}$ and $Tick_{\varepsilon_2}$ applied to some configuration
 have the same effect of the single tick  $Tick_{\varepsilon}$, 
for arbitrary $\varepsilon_1 , \varepsilon_2$ and \,$ \varepsilon = \varepsilon_1 + \varepsilon_2$. 

Indeed, this is a property of the multiset rewriting formalism itself. In this context, above property reflects the continuity of time in the physical world.

With this property in mind, in any trace we can replace consecutive ticks 
$$\Sscr_0 \lra_{Tick_{\varepsilon_1}} \Sscr_1 \lra_{Tick_{\varepsilon_2}} \ \dots \  \lra_{Tick_{\varepsilon_n}} \Sscr_n$$ 
with a single tick
$$\Sscr_0 \lra_{Tick_{(\varepsilon_1 + \varepsilon_2 + \dots + \varepsilon_n)}} \Sscr_n,$$ 
and vice versa,
 without compromising the semantics of the process that is being modelled. 
\end{remark}

\subsection{%Timed
Balanced Systems}
\label{sec:balanced}

The balanced condition~\cite{kanovich11jar} is necessary for decidability of problems such as reachability studied in \cite{kanovich13ic,kanovich.mscs,kanovich15post} as well as the problem introduced in Section~\ref{sec:timedprop}.

 % which is a special kind of balanced timed MSR. Intuitively, progressive timed MSR are such that only a finite number of actions can be carried out in a bounded time interval which is a natural condition for many systems. While the definition of balanced systems is taken from our previous work~\cite{kanovich11jar}, the definition of progressive systems in this setting is new. The balanced condition is necessary for decidability of problems (such as reachability as well as the problems introduced in Section~\ref{sec:timedprop}).

\begin{definition} 
\label{def:balanced}
A timed MSR with dense time $\Tscr$ is \emph{balanced} if for all instantaneous rules $r \in \Tscr$, $r$ creates the same number of facts as it consumes, that is, instantaneous rules (Eq.~\ref{eq:instantaneous-real}) are of the form: %in Eq.~(\ref{eq:instantaneous}), $n = m$.  
\begin{equation}
\begin{array}{l}
 Time@T, \,\Wscr, \red{\,F_1@T_1'}, \ldots, \red{\,F_n@T_n'} \ \mid \  \,\Cscr \ \lra \\ \ \
 \exists \vec{X}.~[ \  Time@T,\, \Wscr, \blue{\,Q_1@(T + D_1)}, \ldots, \blue{\,Q_n@(T + D_n)} \, ] \ ,
\label{eq:instantaneous-balanced}
\end{array}
\end{equation}
where $\Wscr$ is a multiset of timestamped facts.
\end{definition}

\vspace{0,5em}
By consuming and creating facts, rewrite rules can increase and decrease the number of facts in configurations throughout a trace. 
However, in balanced MSR systems, the number of facts in configurations in a trace is constant, as states the following proposition.

\begin{proposition}
  Let $\Tscr$ be a balanced timed MSR with dense time. Let $\Sscr_0$ be an initial configuration with exactly $m$ facts. For all 
%possibly infinite 
traces $\Pscr$ of $\Tscr$ starting with $\Sscr_0$, all configurations $\Sscr_i$ in $\Pscr$ have exactly $m$ facts.
\end{proposition}
 \begin{proof}
Since all the rules in $\Tscr$ are balanced, rule application does not effect the number of facts in a configuration. That is, enabling configuration has the same number of facts as the resulting configuration.
 Hence, throughout the trace, all configurations have the same number of facts as the initial configuration $\Sscr_0$.
\qed
 \end{proof}

% \vspace{-2mm}
\section{Quantitative Temporal Properties}
\label{sec:timedprop}

%\section{%Timed
%From Discrete to Dense Time PTS}
%\label{sec:pts-dense}

\subsection{Goals, Critical Configurations and Non-critical Traces in MSR Systems with  Dense Time}
% and Non-crit Traces}

% \vspace{-2mm}
In order to define quantitative temporal properties, we review the notion of critical configurations and compliant traces from our previous work~\cite{kanovich12rta} and  introduce % relevant  
reachability problem for MSR systems with dense time which considers critical configurations.
%verification problems for  time-sensitive distributed systems that involve  compliant traces for timed MSR models. 

\begin{definition}
\emph{Critical configuration specification} $\CS$
(resp. a \emph{ goal} $\GS$) % configuration specification})  
is a set of pairs 
$$ %\CS =
 \{~\tup{\Sscr_1, \Cscr_1}, \ldots, \tup{\Sscr_n, \Cscr_n}~\} \ .$$
 Each pair ~$\tup{\Sscr_j,\Cscr_j}$~ is of the form:
$$
  \tup{~\{F_1@T_1, \ldots, F_p@T_p\}, \Cscr_j~} 
$$
\noindent
where $T_1, \ldots, T_p$ are time variables, $F_1, \ldots, F_p$ are facts (possibly containing variables) and $\Cscr_j$ is a set of time constraints involving only the variables $T_1, \ldots, T_p$. 

Given a critical configuration specification $\CS$ (resp. a goal $\GS$), % configuration specification),
% $\CS$, 
 we classify a configuration $\Sscr$ as a \emph{critical configuration} w.r.t $\CS$ (resp. \emph{goal} \emph{configuration} w.r.t. $\GS$) 
 if for some $1 \leq i \leq n$, there is a grounding substitution, $\sigma$, 
%  mapping time variables in $\Sscr_i$ to natural numbers, nonce names to nonce names (renaming of nonces) and non time variables to terms
such that:
\begin{itemize} 
   \item $\Sscr_i \sigma \subseteq \Sscr$; 
   \item All constraints in $\Cscr_i \sigma$ are satisfied;
 \end{itemize} 
where substitution application ($\Sscr \sigma$) is defined as usual~\cite{enderton}, \ie, by mapping time variables in $\Sscr$ to natural numbers, nonce names to nonce names (renaming of nonces) and non time variables to terms.
\end{definition}

For simplicity, when the corresponding critical configuration specification or  goal is clear from the context, we will elide it and use terminology {\em critical} or {\em goal configuration}.
% eliding the $\CS$ w.r.t.  which it is defined.

\vspace{5pt}
Notice that nonce renaming is assumed as the particular nonce name should not matter for classifying a configuration as a critical or a goal configuration.
Nonce names cannot be specified in advance, since these are freshly generated in a trace, \ie~during the execution of the process being modelled.

%\subsubsection{Non-critical Traces in Dense Timed Systems}
Moving from discrete to dense time % in the planning problem 
is not straightforward 
%when considering
w.r.t. the notion of a compliant, \ie, non-critical trace.
 Consider, for example, a  trace in a  timed MSR with dense time, containing the following configurations and a $Tick$:
$$
 Time@1.5,  \, F@3.5 \ \lra_{Tick_{3}} \ Time@4.5,  \,F@3.5 % \ \lra_{r_{2}} \  Time@(T+3),  C@(T+3) \ . 
$$
which could potentially be considered as non-critcal w.r.t. with the   critical configuration specification:
$$
Time@T,  \,F@T_1  \ \mid \ \{\, T_1 =T \, \}\ ,
$$
as it doesn't contain any critical configurations.
However, a trace containing rules:
$$
\begin{array}{lc}
Time@1.5, \, F@3.5  \ \lra_{Tick_{2}} \ Time@3.5,  \,F@3.5 \ 
\lra_{Tick_{1}} \ Time@4.5,  \,F@3.5 % \ \lra_{r_{2}} \  Time@(T+3),  C@(T+3)
\end{array}
$$
would not be non-critical w.r.t. the same   critical configuration specification since it contains the critical configuration \ $\{ \,Time@3.5, \, F@3.5\, \} $.
% \ which is critical. 
Above traces differ only in the representation of  time flow and they model the same real-time process.
In reality, due to continuity of time, the  process would reach such a critical state, \ie~it would not skip over this undesired state.
Clearly, this inconsistency is not what we want in our model.

\vspace{0,5em}
As the above example suggests,  in the setting with dense time it is particularly important that the notion of a non-critical trace is properly defined.
While in systems with discrete time, time can increase only by one time unit at a time, when time is dense, time can increase by any value, however small, and however large. %  and that's what we want in our system.
That is how we model the natural continuous aspect of time we know in our everyday life. In particular, recall Remark~\ref{remark-consecutive-ticks}, illustrating how the continuity of time flow is implicitly embedded in the MSR formalism.
Namely,   given arbitrary $\varepsilon>0$ and any positive $\varepsilon_1 < \varepsilon $, there exists $\varepsilon_2>0$  such that the time $Tick$ for $\varepsilon$ has the same effect as the  $Tick$ for $\varepsilon_1$ followed by the $Tick$ for $\varepsilon_2$. That is, if \
$$\Sscr_0 \lra_{Tick_{\varepsilon}} \Sscr_1$$  \ then \ 
  $$\Sscr_0 \lra_{Tick_{\varepsilon_1}} \Sscr_2 \lra_{Tick_{\varepsilon_2}} \Sscr_1\ . $$ 
Clearly, $ \varepsilon= \varepsilon_1 + \varepsilon_2$ holds.
%Above property is already naturally embedded in the multiset rewriting formalism. 
Relying on above property, we now % carefully 
define which  traces  may be considered as compliant in the dense time setting.

\begin{definition}
\label{def:compliant}
 Given a %planning problem 
timed MSR with dense time $\Tscr$ and a critical configuration specification $\Cscr\Sscr$,
%and an initial configuration $\Sscr_I$,
%a goal $\Sscr_G$, a finite set of critical configurations $\Cscr$  and a set of actions $\Rscr$, a plan $\Sscr_I \lra_\Rscr^* \Sscr_G$
a trace $\Pscr$ of $\Tscr $  is \emph{non-critical} if 
%it does not reach any critical configuration, and if
 no  critical configuration is reached along any trace obtained by replacing  any subtrace 
$$\Sscr_i \lra_{Tick_{\varepsilon}} \Sscr_{i+1}$$
 of \ $\Pscr$ 
%$\Sscr_I \lra_\Rscr^* \Sscr_G$  
with \qquad \quad \qquad\qquad   $\Sscr_i \lra_{Tick_{\varepsilon_1}} \Sscr' \lra_{Tick_{\varepsilon_2}} \Sscr_{i+1} $
\\[6pt]
for arbitrary $\varepsilon_1 < \varepsilon$, such that \,   $ \varepsilon = \varepsilon_1 + \varepsilon_2$ holds. 
\end{definition}

 Above decomposition of the $Tick$ rules, in all possible ways of consecutive $Tick$s, ensures that the continuity of time and the notion of non-critical traces are well combined.

On the other hand, however, checking whether a given trace in a system with dense time is non-critical  is potentially more challenging than in the untimed setting~\cite{kanovich11jar} and models with discrete time~\cite{kanovich12rta,kanovich16formats}. 
Testing whether a trace is non-critical  in models with dense time requires potentially checking through an infinite number of traces. This could possibly effect the complexity  of the corresponding  non-critical reachability problem.
Fortunately, we can rely on our equivalence relation among configurations, \ie~on our technical machinery called  circle-configurations, with respect to this issue as well. 
We  show this result  in Section~\ref{sec: compliance}.

% Notice that any finite trace $\Pscr$ of a timed MSR $\Ascr$ can be the prefix of a infinite trace of $\Ascr$, namely, the following one:
% \[
%   \Pscr \lra_{Tick} \Sscr_1 \lra_{Tick} \Sscr_2 \lra_{Tick} \cdots 
% \]

%\vspace{-2mm}
\subsection{Verification Problem}
%\vspace{-2mm}

\begin{definition}
\label{def:non-critical-trace}[Non-critical reachability problem]\ \\
 Given a timed MSR $\Tscr$,  a goal %configuration specification 
$\GS$,  a critical configuration specification $\CS$ and an initial configuration $\Sscr_0$,
is  there a non-critical trace, $\Pscr$, that leads from $\Sscr_0$  to a goal configuration? % $\Sscr_G$?
 \end{definition}

%As for discrete time setting (Section~\ref{sec:complex}), 
Our complexity results,  for a given MSR  $\Tscr$, an  initial configuration $\Sscr_0$, a critical configuration specification $\CS$ and a goal %configuration specification
 $\GS$, mention the value $\Dmax$ which is an  upper-bound on the natural numbers appearing in $\Sscr_0$, $\Tscr$, $\CS$ and $\GS$, which is syntactically inferred from timestamps and numbers appearing in facts, rules and constraints of $\Sscr_0$,  $\Tscr$, $\CS$ and $\GS$.

\vspace{5pt}
For the complexity results for non-critical reachability problem (bisimulation of non-critical traces) with dense time  we define {immediate successors} for configurations, motivated by the non-determinism in the model related to the choice of the positive real number $\varepsilon$ used in the $Tick$ rule.
Namely,  unless some restrictions are imposed on a trace by some time sampling, $Tick_{\varepsilon}$ rule is applicable to every configuration, and for every $\varepsilon >0$. However, the choice of $\varepsilon$ is important as it may have %undesired 
different effects on representation of time in a trace.
Consider, for example, configuration
$$
\Sscr = \{ \,Time@2, \,F@0.4, \,G@2.5, \,H@1\,\} \ .
$$  
Applying a $Tick$ rule to $\Sscr$ for any $\varepsilon<0.4$ has the same effect w.r.t time constraints satisfied by the resulting configuration, regardless of a particular  $\varepsilon<0.4$ used. In fact, it has no effect in that sense,
% the same effect as not advancing the global time, 
since the same set of constraints is satisfied by the resulting configuration as by configuration $\Sscr$.
Advancing time in $\Sscr$ by $\varepsilon = 0.4$ is different. Resulting configuration \ $$
\Sscr' = \{ \,Time@2.4, \,F@0.4, \,G@2.5,\, H@1\,\} \ .
$$
satisfies \eg~constraint $T' =T$, related to facts $Time@T$ and $G@T'$, which is not satisfied by $\Sscr$. 
Now, applying a $Tick_{\varepsilon}$ to $\Sscr'$ for any $\varepsilon >0$ would change the set of constraints satisfied by the resulting configuration $\Sscr''$. 
The set of constraints satisfied by  $\Sscr''$ will depend on the value of $\varepsilon$.
For example, for \ $\varepsilon = 0.35$ constraint  $T>T'+2$, where $T, T'$ relate to facts $Time@T$ and $F@T'$,  would be satisfied in $\Sscr''$ ($2.45>0.4+2$) and $T=T'+2$ would not, while for $\varepsilon=0.3, $ constraint $T= T'+2$ would hold.
%$\varepsilon=1.3, $ \eg~ constraint $T= T'+3$ would hold ($2.1+1.3=3.4 =0.4+3 $).

\vspace{0,3em}
With the above consideration on the importance on how much the time advances by a single $Tick$ rule, we define the following, {\em successor},  relation among configurations.
% It will be used in the definition of new  time sampling for models with dense time.

\begin{definition}
\label{def: succesor conf}
Given a timed MSR $\Tscr$ with dense time, and a natural number $d$, % $\Dmax$,
 let  ~$\Cscr_d$~ %$\Cscr_{\Dmax}$ 
be a set of all constrains containing natural numbers up to $d$: % $\Dmax$:
$$
\Cscr_d %{\Dmax}
 = \{ \  T > T' \pm N, \ T \geq T' \pm N , \   T = T' \pm N \ \mid  \ N \leq \ d \ \}
$$
%a set of constraints $\Cscr$
%critical configuration specification $\CS$ and an initial configuration $\Sscr_0$ 
We say that configuration $\Sscr_2$ is an \emph{immediate successor} of configuration  $\Sscr_1$ w.r.t. $d$ if \   
\begin{itemize}
\item[i)] There exists $\varepsilon>0$ such that \  $\Sscr_1 \lra_{Tick_{\varepsilon}}\Sscr_2$;
\item[ii)]  $\Sscr_1$ and  $\Sscr_2$ \ do not satisfy the same set of constraints from $\Cscr_d$, where variables $T$ and $T'$ refer to timestamps of same facts from $ \Sscr_1$ and $ \Sscr_2$;
\item[iii)] For all $ \varepsilon' >0, \ \varepsilon'<\varepsilon$ \ if \ $ \Sscr_1 \lra_{Tick_{\varepsilon'}} \Sscr'$  then $\Sscr'$ satisfies the same constraints from $\Cscr_d$ either as $\Sscr_1$ or as $\Sscr_2$.
\end{itemize}
When  $\Sscr_2$ is an {immediate successor} of   $\Sscr_1$   w.r.t. $d$
we write\ \ \ $\Sscr_1 \lra_{Tick_{IS}^d}\Sscr_2$ . %\  or simply \  $\Sscr_1 \lra_{Tick_{IS}}\Sscr_2$.
\end{definition}

When $d$ is clear from the context we simply say that $\Sscr_2$ is an {immediate successor} of   $\Sscr_1$   
and write \ $\Sscr_1 \lra_{Tick_{IS}}\Sscr_2$.

Notice that in the above example,  $\{ \,Time@2.05, \,F@0.4, \,G@2.1, \,H@1\,\}$  is an immediate successor of $\Sscr$, while configuration 
$
\{\, Time@2.4, \,F@0.4, \,G@2.5, \,H@1\,\} 
$
is not because, \eg
$$
%\begin{small}
\begin{array} {l}
\{ \,Time@2, \,F@0.4, \,G@2.5, \,H@1\,\}\, \lra_{Tick_{0.05}} \\[5pt] \quad 
\{ \,Time@2.05, \,F@0.4, \,G@2.5, \,H@1\,\}\, 
 \lra_{Tick_{0.35}}  \\[5pt] \qquad 
\{ \,Time@2.4, \,F@0.4, \, G@2.5, \,H@1\,\} \ 
\end{array}
%\end{small}
$$
where all of the above configurations satisfy different time constraints.

In general, the immediate successor of a configuration is not unique. For example, %both configurations
 $\{ \,Time@2.15, \,F@0.4, \,G@2.5, \,H@1\,\}$ and $ \{ \,Time@2.3, \,F@0.4, \,G@2.5, \,H@1\,\}$ are both  immediate successors of $\Sscr'$.
On the other hand, the immediate successor of % configuration  
$\{ \,Time@2.15,\, F@0.4,\, G@2.5, \,H@1\,\}$   is unique,  $\{\, Time@2.4, \,F@0.4, \,G@2.5,\, H@1\,\}$.

%\subsection{Non-critical Traces %Compliance 
%in Dense Time Systems}

\vspace{1em}

There is a clear connection between non-critical traces and immediate successor configurations.
Notice that if  neither  $ \Sscr_i$ nor  its immediate successor configuration $ \Sscr_{i+1}$ is critical, then the condition on non-critical traces given in Definition~\ref{def:compliant}  is satisfied. 

\begin{proposition}
\label{th: IS critical}
Let $\Tscr$ be a timed MSR with dense time, and  $d$ %$\Dmax$ 
 a natural number. 
%such that it is greater than all of the  numeric values appearing in specification of $\Tscr$ and $\CS$.
Let \ $\Sscr \lra_{Tick_{\varepsilon_{IS}}} \Sscr'$.
If \  $ \Sscr$  and $ \Sscr'$ are not critical w.r.t. some critical configuration specification $\CS$ involving constraints  form $\Cscr_{d}$, 
then for any \ $ \varepsilon' >0$, $\varepsilon' <\varepsilon$, \ the configuration \ $\Sscr''$ such that  \
$\Sscr \lra_{Tick_{\varepsilon_1}} \Sscr'' \lra_{Tick_{\varepsilon_2}} \Sscr'$,   is not critical.
\end{proposition}
\begin{proof}
Let $\Sscr \lra_{Tick_{\varepsilon_{IS}}} \Sscr'$, \ and assume neither $\Sscr$ nor $\Sscr'$  is critical. Let 
$$\Sscr \lra_{Tick} \Sscr'' \lra_{Tick} \Sscr' \ . $$

%As per definition  (Definition~\ref{def:compliant}), we need to show that  if \ $\Sscr  \to_{Tick_{IS}}  \Sscr'$, then no critical configuration $\Sscr''$  exists such that 

Since $ \Sscr'$ is an  immediate successor of $ \Sscr$, as per Definition~\ref{def: succesor conf},   such configuration $\Sscr''$ satisfies the same set of constraints  form $\Cscr_{d}$ as either $\Sscr$ or $\Sscr'$.
This includes the constrains used in $\CS$.
 Since both $\Sscr$ and $\Sscr'$ are not critical, $\Sscr''$ is not critical as well.
\qed
\end{proof}

% Namely, \ for any $ \varepsilon' >0$, $\varepsilon' <\varepsilon$, \ $\Sscr'$ in \
%$\Sscr_i \lra_{Tick_{\varepsilon_1}} \Sscr' \lra_{Tick_{\varepsilon_2}} \Sscr_{i+1}$,   satisfies the same constraints as either $ \Sscr_i$ or $ \Sscr_{i+1}$.  
%In other words, $ \Sscr'$ is equivalent to either $ \Sscr_{i}$ or $ \Sscr_{i+1}$.

%\newpage

\section{Complexity Results for Balanced Timed MSR  with Dense Time}
\label{sec:complex-real}

%  by simply inspecting the timestamps of $\Sscr_0$, the $D$ values in timestamps of rules (which are of the form \ $T + D$) and constraints in $\Ascr$ and $\CS$ (which are of the form $T_1 > T_2 \pm D$, $T_1 = T_2 \pm D$ and\, $T_1  \geq T_2 \pm D$). For example, the $\Dmax = 1$ for the specification in Figure~\ref{fig:rules-complete}.

Reachability and the related problems for MSR are undecidable in general~\cite{kanovich09csf}. However, by imposing some restrictions on the form of the rewrite rules, such as using only balanced rules and bounding the size of facts, these problems become decidable, even in timed models with fresh values.
 
A summary of related complexity results in shown in Table \ref{table:summary TR}.

\begin{table}[t]
\caption{Summary of the complexity results for the reachability and non-critical reachability problems. 
These results also hold for  MSR models with fresh values.
}
\label{table:summary TR}

\begin{center}
    \begin{tabular}{ |@{\quad} c@{\quad}  | c|  c | c | } %p{3,4cm} |   p{3,4cm} |}
\hline
\multicolumn{2}{ |c| }{MSR}
&  {\bf \ \ Reachability Problem \ \ } &{\bf \ \ Non-critical Reachability \ \ }  \\[0,1em]    \hline
        \multirow{6}{*}{Balanced} &   \multirow{2 }{*}{untimed}  & PSPACE-complete &  PSPACE-complete \\ 
& & ~\cite{kanovich11jar,kanovich13ic}& ~\cite{kanovich11jar,kanovich13ic}
\\   \cline{2-4}
  \multirow{2 }{*} &  \multirow{2 }{*}{\ \ discrete time \ \ }  & PSPACE-complete &  PSPACE-complete \\ 
 & &  ~\cite{kanovich17mscs} &  ~\cite{kanovich17mscs}
 \\ \cline{2-4}
 \multirow{2 }{*} &  \multirow{2 }{*}{real time}  & PSPACE-complete &  {\bf \textcolor{red}{PSPACE-complete}}
\\ 
& &  ~\cite{kanovichJCS17} & {\bf \textcolor{red}{new!}}
\\   \hline
\multicolumn{2}{ |c| }{\multirow{2}{*}{\ \ Not necessarily balanced \ \ } } & % \multicolumn{1}{ c| }
{Undecidable} & {Undecidable}  \\ 
\multicolumn{2}{ |c| }{}& ~\cite{kanovich09csf} &  ~\cite{kanovich09csf} \\
\hline
    \end{tabular}

\normalsize
\end{center}

\end{table}

In this section we investigate the complexity of  the non-critical reachability problem for balanced systems with facts of bounded size.

%\vspace{0,5em}
In this new setting with dense time, %we want to investigate the complexity of 
the non-critical reachability problem %. This problem 
combines   quantitative temporal properties defined for timed MSR  with the refined notion of compliance.
%, \ie, non-critical traces for dense time setting.
Our results rely heavily on the abstractions called circle-configurations.
As we will show in  Section~\ref{sec: compliance},
 circle-configurations and the related time advancement rules, $Next$,  are defined in such a way to reflect  similar characteristics related to advancement of time in dense time models.

\vspace{3pt}
As discussed above, we  assume a bound, $k$, on the size of facts. However,  we do not  impose an upper bound on the values of timestamps. 
Also, our  timed MSRs with dense time are constructed over $\Sigma$, a finite alphabet with $J$ predicate symbols and $E$ constant and function symbols and can involve an unbounded number of fresh values.

\subsection{Circle-configurations}
\label{sec: compliance}

In order to handle dense time, and in particular for our complexity results, in our previous work~\cite{kanovich15post} we introduced an equivalence relation among configurations. 
We now review main ideas behind this machinery. For a more detailed exposition of this approach  see~\cite{kanovich15post}.

The equivalence of  configurations involves  an
upper bound $\Dmax$ on the numeric values mentioned in the specification of the considered system and problems:
% of a reachability problem. This value is  computed from the given problem: 
We set $\Dmax$ to be a
% the least
 natural number such that \mbox{$\Dmax > n + 1$} for any number $n$ (both real or natural) appearing in the timestamps of the 
initial configuration, or the $N$s and $D_i$s in constraints (Eq.\ref{eq:constraints}) or rules (Eq.\ref{eq:instantaneous-real}) of the timed MSR, in goal and critical  configuration specification.
%reachability problem. 

Notice that immediate successor configurations also involve an upper bound, $d$, on natural numbers appearing in time constraints.
For a given problem, we will extract the value $\Dmax$ as described above, and we will consider immediate successor configurations w.r.t. the same bound $\Dmax$.

Configurations 
are defined as equivalent if   they contain the same  (untimed) facts, up to nonce renaming, and if they satisfy the exact same set of  constraints.
When we say that some configurations  satisfy the same constraint, we  intend to say that time variables of that constraint refer to the same facts in both configurations.

\begin{definition}
\label{def:equivalence}
 Given a % reachability problem $\Tscr$,
timed MSR $\Tscr$ with dense time, a goal $\GS$, a critical configuration specification $\CS$ and an initial configuration $\Sscr_0$,
 let $\Dmax$ be an upper bound on the numeric values
 appearing in $\Tscr$,  $\GS$, $\Cscr\Sscr$ and $\Sscr_0$.
Let
\begin{equation}
\label{eq:two-configurations}
\begin{array}{l}
\Sscr = \{\,  Q_1@t_1,\,  Q_2@t_2, \ldots,\,  Q_n@t_n \ \}
\quad \textrm{and} \ \quad
\widetilde{\Sscr} =  \{ \ \widetilde{Q}_1@\widetilde{t}_1, \, \widetilde{Q}_2@\widetilde{t}_2, \ldots,
\, \widetilde{Q}_n@\widetilde{t}_n \ \}
\end{array} 
\end{equation}
be two configurations written in canonical way where the two sequences of
timestamps \ $t_1, \ldots, t_n$ \ and \ $\widetilde{t}_1, \ldots, \widetilde{t}_n$ \ are non-decreasing. 
(For the case of equal timestamps, we sort the facts in alphabetical order, if necessary.) 
We say that  configurations  $\Sscr$
and $\widetilde{\Sscr}$ are \emph{equivalent configurations}
% and write\  $\Sscr \sim_{\Dmax} \widetilde{\Sscr}$ or simply  $\Sscr \sim \widetilde{\Sscr}$ , 
\ if the following conditions hold:
\begin{itemize}
\item[(i)]
There is a bijection $\sigma$ that maps the set of all nonce names appearing in 
%one 
configuration $\Sscr$ to the set of all nonce names appearing in
% the other 
configuration $\widetilde{\Sscr}$,
such that \ $Q_i \sigma = \widetilde{Q}_i$, \ for each $i \in \{1,\dots, n\}$;  and
\item[(ii)]  Configurations ${\Sscr}$ and $\widetilde{\Sscr}$ satisfy the same constraints, that is:
$$
\begin{array}{l}
t_i > t_j \pm D \quad  \textrm{ iff } \quad \widetilde{t_i} > \widetilde{t_j} \pm D \quad  \ \text{and} \\
 t_i = t_j \pm D \quad  \textrm{ iff } \quad \widetilde{t_i} = \widetilde{t_j} \pm D ,
\end{array}
$$
for all \ $1\leq i \leq n$, $1\leq j \leq n$ \  and \ $D \leq \Dmax$.
\end{itemize}
When    $\Sscr$ and $\widetilde{\Sscr}$ are equivalent we  write\  $\Sscr \sim_{\Dmax} \widetilde{\Sscr}$, or simply  $\Sscr \sim \widetilde{\Sscr}$.
\end{definition}

As we already pointed out , when we say that  $\Sscr$ and $\widetilde{\Sscr}$ satisfy the same constraints, we mean that the time variables in the constraint refer to the same facts $Q_i$ and $\widetilde{Q}_i$, up to nonce renaming.

Notice that  no configuration is equivalent to its immediate successor configuration.

\vspace{0,5em}
In~\cite{kanovich15post} we also introduced an illustrative representation of the above equivalence relation,  called \emph{circle-configuration}.

\begin{definition}
\label{def:circle-conf}
Let $\Tscr$ be a timed MSR with dense time,  $\GS$ a goal, $\CS$  a critical configuration specification and  $\Sscr_0$ an initial configuration.
Let $\Dmax$ be an upper bound on the numeric values  appearing in $\Tscr$,  $\GS$, $\Cscr\Sscr$ and $\Sscr_0$,  and \ %$\Sscr$ be the configuration \ 
%\begin{equation}
%\begin{array}{l}
$\Sscr = F_1@t_1, F_2@t_2, \ldots, F_n@t_n, Time@t \ . $
%\end{array} 
%\end{equation}
%be a configurations written in canonical way where the sequence of timestamps \ $t_1, \ldots, t_n$ \ is non-decreasing. 
The pair\  $ \Ascr_{\Sscr}=\tup{\Delta_{\Sscr}, \Uscr_{\Sscr}}$  \ is the  \emph{circle-configuration} of the configuration $\Sscr$ defined as follows. 
The  \emph{$\delta$-configuration} of
% the configuration
 $\Sscr$, $\Delta_{\Sscr}$, is:
\[
 \Delta_{\Sscr} = \left\langle
 \begin{array}{l}
 \{P_{1}^1, \ldots, P_{m_1}^1\}, \delta_{1, 2}, \{P_{1}^2, \ldots, P_{m_2}^2\}, \delta_{2,3}, \ldots 
%  \{P_{1}^{j-1}, \ldots, P_{m_{j-1}}^{j-1}\}, 
, \delta_{j-1, j}, \{P_{1}^j, \ldots, P_{m_j}^j\}
 \end{array}\right\rangle
\]
 where \
% \begin{itemize}\item[(i)]
 $\{P_1^1, \ldots, P_{m_1}^1, P_1^2, \ldots, P_{m_j}^j\} = \{F_1, \ldots, F_n, Time\}$,
% \item[(ii)]
\ timestamps of facts \ \mbox{$P_1^i, \ldots, P_{m_i}^i$} \ have the same integer part, $t^{\, i}$,  $\forall i = 1, \dots, j$ ,
and 
%\item[(iii)]
 $$ 
\delta_{i,i+1} =% \delta_{Q_{1}^i, Q_{1}^{i+1}}= 
\left\{\begin{array}{l}
{t^{\, i+1}} - t^{\, i}, \textrm{ if } \ {t^{\, i+1}} - t^{\, i} \leq \Dmax\\
% \Int{t_{i+1}} - \Int{t_i}, \textrm{ if } \ \ \Int{t_{i+1}} - \Int{t_i} \leq \Dmax\\
 \infty, \textrm{ otherwise }
\end{array}\right.
\ , \quad  i = 1, \dots, j-1 \ .
$$
% is the truncated time difference between the facts in class $i$ and class $i+1$.
%\end{itemize}
The  \emph{unit circle} of
% the configuration 
$\Sscr$, \ $\Uscr_{\Sscr}$, is:
$$
\Uscr_{\Sscr} = [ \ \{Q_{1}^0, \ldots, Q_{m_0}^0\}_\Zscr, \{Q_{1}^1, \ldots, Q_{m_1}^1\}, \ldots, \{Q_{1}^k, \ldots, Q_{m_k}^k\} \ ]
$$
 where \
$\{Q_1^0, \ldots, Q_{m_0}^0, Q_1^1, \dots, Q_{m_k}^k\} = \{F_1, \ldots, F_n, Time\}$,  \
timestamps of facts in the same class, \ $Q_1^i, \ldots, Q_{m_i}^i$ \ have the same decimal part,  $\forall i = 0, \dots, k$ ,
 timestamps of facts \ $Q_1^0, \ldots, Q_{m_0}^0$ \ are integers, and 
the classes are ordered in the increasing order, \ie, $\dec{Q_i^l}<\dec{Q_j^{l'}}$ for all $i\neq j$, where 
 $1\leq i \leq m_l$, $1 \leq j \leq m_{l'}$, $0\leq l\leq k$, $1 \leq l'\leq k$.
%
%We say that $\Delta$  is the \emph{$\delta$-configuration} of
%% the configuration
% $\Sscr$,
% \ $\Uscr$ is the \emph{unit-circle} of
%% the configuration 
%$\Sscr$, 
%and \ $\tup{\Delta, \Uscr}$  \ is the  \emph{circle-configuration} of the configuration $\Sscr$ (or the circle-configuration corresponding to %$\Sscr$). 
%% where $\Delta$ is the $\delta$-configuration of $\Sscr$ and $\Uscr$ is the unit circle of $\Sscr$.

We write \  $\Uscr_{\Sscr}(Q_{j}^i) = i$ \ to denote the class in which the fact $Q_{j}^i$ appears in $\Uscr_{\Sscr}$. 

\end{definition}

\begin{figure}[t]
\centering
\begin{minipage}{.45\textwidth}
  \centering
\vspace{1em}
\includegraphics[width=4.6cm]{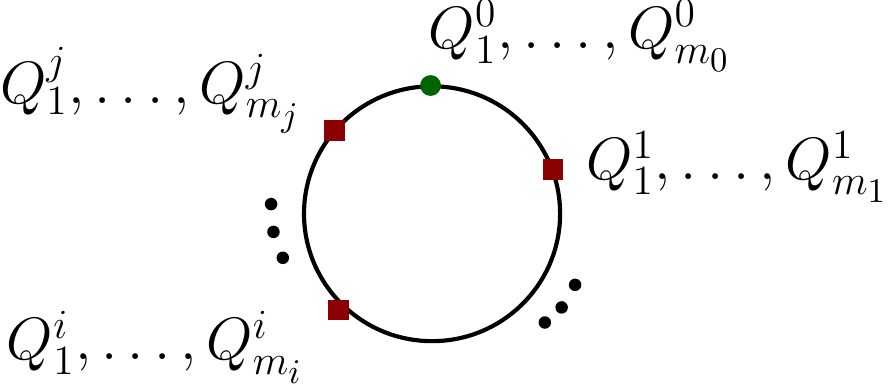}
\vspace{1.1em}
  \captionof{figure}{Unit Circle}
  \label{fig:unit-graph}
\end{minipage}%
\begin{minipage}{.55\textwidth}
  \centering
  \includegraphics[width=5.6cm]{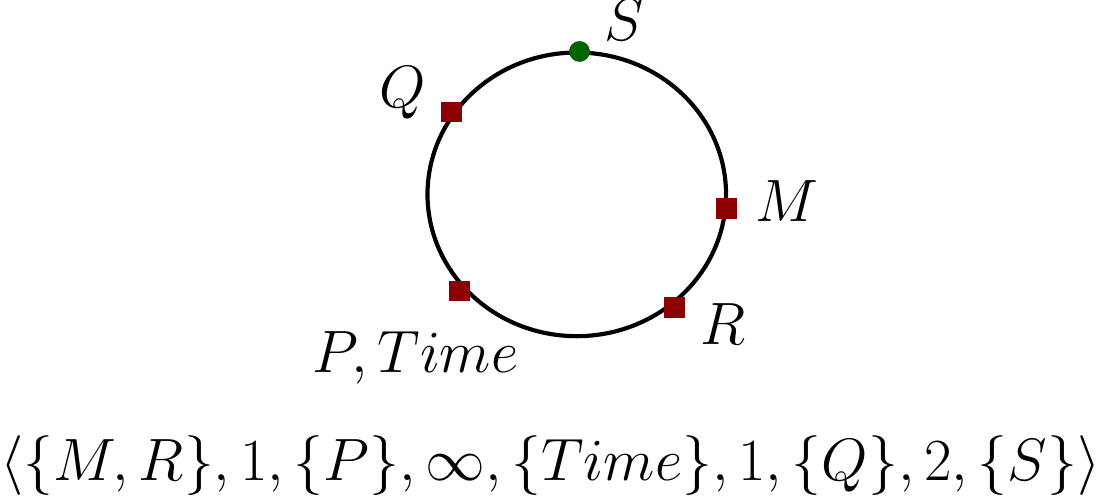}
  \captionof{figure}{Circle-Configuration  }
  \label{fig:circle-graph}
\end{minipage}
\vspace{-1em}
\end{figure}

For simplicity, we sometimes write $\Ascr$ and $\tup{\Delta, \Uscr}$ instead of  $\Ascr_{\Sscr}$ and $\tup{\Delta_{\Sscr}, \Uscr_{\Sscr}}$, when the corresponding configuration is clear from the context.

\vspace{0,5em}
We graphically represent a unit circle as shown in Figure~\ref{fig:unit-graph}.
The class marked with the subscript $\Zscr$,  $\{Q_{1}^0, \ldots, Q_{m_0}^0\}_\Zscr$, is called  the \emph{zero point} and is marked as the (green) ellipse at the top of the circle.
The remaining classes are placed on the circle as 
the (red) squares ordered clockwise starting from the zero point. From the above graphical representation, given  in Figure~\ref{fig:unit-graph},  it can easily be seen that the decimal part of the timestamp of the fact $Q_{1}^1$ 
is smaller than the decimal of the timestamp of the fact $Q_{1}^2$, while the decimal part of the 
timestamps of the facts $Q_{1}^i$ and $Q_{2}^i$ are equal. 
The exact points where the classes  are 
placed on the circle are not important, only their relative positions matter.
As an example, the circle-configuration of configuration \

\vspace{2pt}
\qquad 
$ \{ \ M@3.01, \,R@3.11, \,P@4.12, \,Time@11.12, \,Q@12.58, \,S@14\ \}$
\\[3pt]
for~$\Dmax = 3$~ consists of 
the $\delta$-configuration \ 
$$\Delta_{\Sscr_1} = \left\langle  \  \{ M, R\}, 1, \{ P\}, \infty, \{Time\}, 1, \{Q\}, 2, \{ S\}  \ \right\rangle$$
and  the unit circle \  
$$  [\ \{S\}_\Zscr, \{M\}, \{R\}, \{P, Time\},\{Q\} \ ],$$
 as illustrated in Figure~\ref{fig:circle-graph}.
%Given a unit circle, $\Uscr$, we write \  $\Uscr(Q_{j}^i) = i$ \ to denote the class in which the fact $Q_{j}^i$ appears in $\Uscr$. 

Notice that, although the graphical representation of the circle-configuration is very illustrative, a circle-configuration is given as a pair of sequences containing a finite number of symbols. Although these sequences do not contain any real numbers, they provide enough information related to satisfaction of  time  constraints, which is necessary \eg~for rule application. 
Circle-configurations are, hence, an elegant representation of  configurations, considering  that timestamps range over dense, real time domain and  that there is no  upper bound on the values of timestamps.

 \vspace{0,5em}
When compared to the equivalence relation between configurations (Definition \ref{def:equivalence}), circle-configurations contain an additional bit of information. While for the equivalence relation only relative differences between concrete values of timestamps of facts are important,  because of the zero point on the unit circle, circle-configurations may differentiate configurations based on the decimal part of their timestamps.
For example, configurations 
\ $ \{ \,Time@1, \,Q@1.54, \,S@2.4\ \} \ \ \text{and} \ \ \{ \,Time@1.12, \,Q@1.66, \,S@2.52\ \}
$
are equivalent, but have different unit circles, related only to the placement of facts at the zero point.

In~\cite{kanovich15post} we have  shown how the notion of circle-configurations corresponds to equivalence relation between configurations.
In particular, configurations corresponding  to the same circle-configuration are equivalent.
%  in the following sense: configurations %$\Sscr$ and $\Sscr'$  are equivalent if and only if they correspond to the same circle-configuration.% $\tup{\Delta, \Uscr}$.
We are, therefore, able to say that a circle-configuration  $\tup{\Delta, \Uscr}$ corresponding to a configuration $\Sscr$ {satisfies a constraint} $c$ if the configuration $\Sscr$ satisfies constraint $c$. We
 %also say that  a  circle-configuration  is a \emph{goal circle-configuration} if it corresponds to a goal configuration. Furthermore, we 
also say that a {rule is applicable to a  circle-configuration} if that rule is applicable to the corresponding configuration.
%
%Recall that equivalent configurations  satisfy the same set of constraints, hence, for equivalent configurations it holds that one %configuration % of them 
%is critical if and only if the other one is critical as well. Therefore, %as with goal circle-configurations,
 Furthermore, we say that a  \emph{circle-configuration is critical} \ iff it is the circle-con\-fi\-gu\-ra\-tion of a critical  configuration. 
%Critical circle-configuration specification 
Analogously,  we say that a  circle-configuration is a \emph{goal circle-configuration} \ iff it is the circle-configuration of a goal configuration.

\vspace{0,5em}
In~\cite{kanovich15post} we show in detail how both instantaneous rules and the time advancement over circle-configurations are compiled and  applied (for more details see~\cite[Section 4.2]{kanovich15post}).
 % on a multiset of circle-configuration.
For an instantaneous rule $r$, we write  $[r]$ for the corresponding rewrite rule  over circle-configurations.

\newcommand\figspace{\vspace{4pt}}

\begin{figure}[t]
\textbullet\quad\textbf{\small Time in the zero point and not in the last class in the unit circle, where $n \geq 0$:}
\figspace

% \begin{center}
\includegraphics[width=8cm]{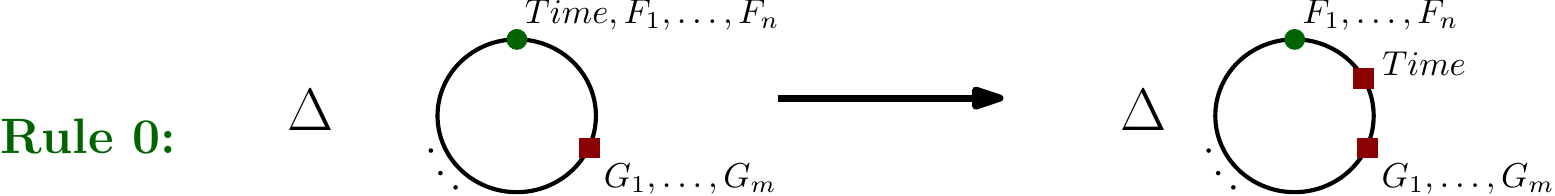}
% \end{center}
\figspace

\textbullet\quad\textbf{\small Time alone and not in the zero point nor in the last class in the unit circle:}
\figspace

% \begin{center}
\includegraphics[width=8cm]{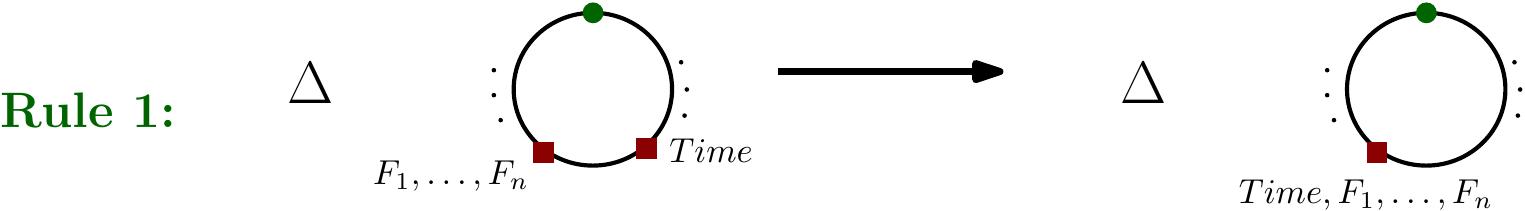}
% \end{center}
\figspace

\textbullet\quad\textbf{\small Time not alone and not in the zero point nor in the last class in the unit circle:}
\figspace

% \begin{center}
\includegraphics[width=10cm]{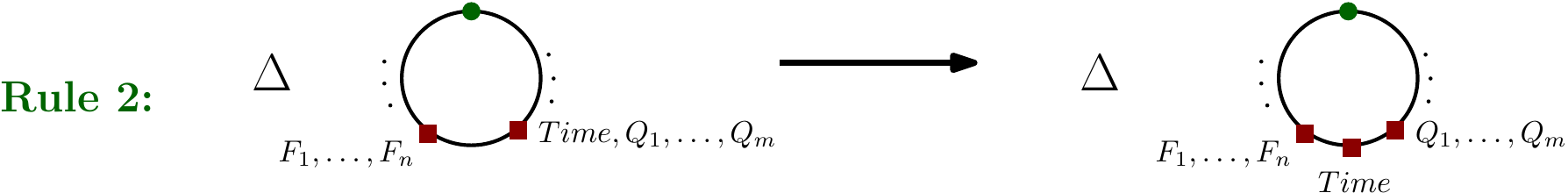}
% \end{center}
\figspace
%\vspace{-1mm}

\textbullet\quad\textbf{\small Time not alone and in the last class in the unit circle which may be at the zero point:}
\figspace

% \begin{center}
\includegraphics[width=9cm]{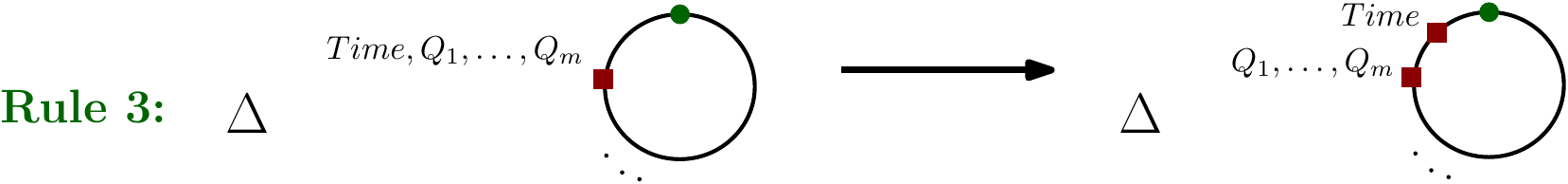}
% \end{center}
\figspace

 \caption{Rewrite Rules for Time Advancement using Circle-Configurations.
%  Here $\Delta$ is the $\delta$-configuration component of  the given Circle-Configuration.
 }
 \label{fig:time-advance}
% \vspace{1em}
\end{figure}

\begin{figure}[t]
\textbullet\quad\textbf{\small Time alone and in the last class in cnit circle - Case 1: $m > 0, k \geq 0, n\geq 0$ and $\delta_1 > 1$:}
\figspace

% \begin{center}
\includegraphics[width=12cm]{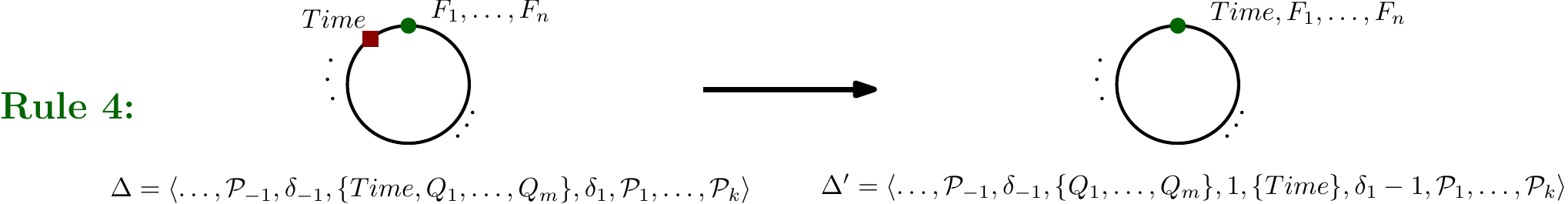}
% \end{center}
\figspace

\textbullet\quad\textbf{\small Time alone and in the last class in unit circle - Case 2: $m > 0, k \geq 1$ and $n\geq 0$:}
% \begin{center}
\figspace

\includegraphics[width=11cm]{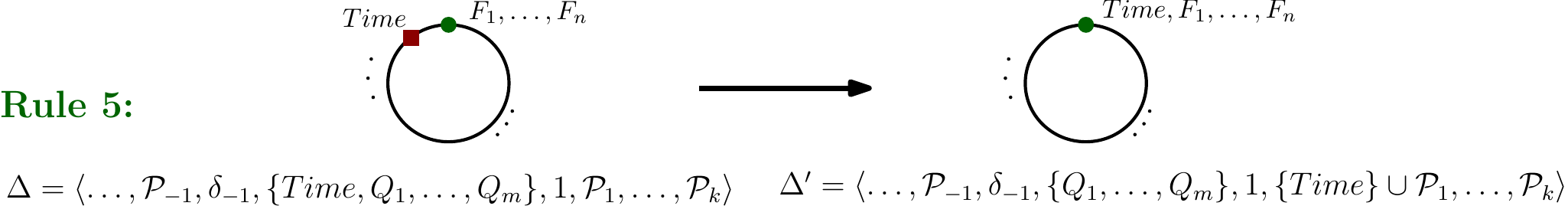}
% \end{center}
\figspace

\textbullet\quad\textbf{\small Time alone and in the last class in unit circle - Case 3: $k \geq 0$ such that $\delta_{1} > 1$ when $k > 0$ and $\gamma_{-1}$
is the truncated time of $\delta_{-1} + 1$:}
% \begin{center}
\figspace

\includegraphics[width=11cm]{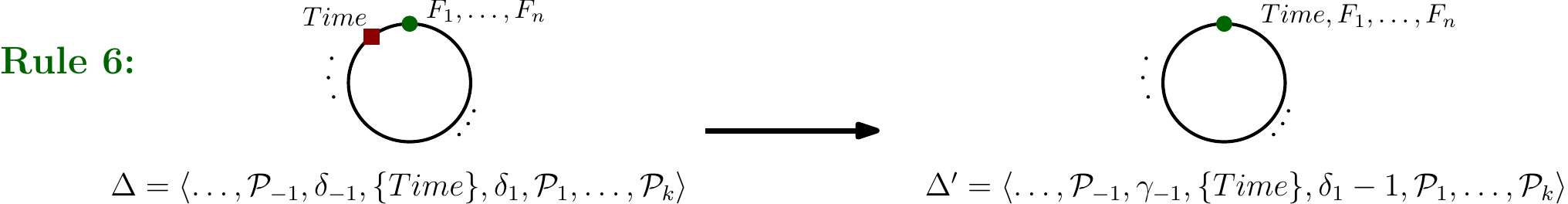}
% \end{center}
\figspace

\textbullet\quad\textbf{\small Time alone and in the last class in unit circle - Case 4: $k \geq 1$ and $\gamma_{-1}$
is the truncated time of $\delta_{-1} + 1$:}
% \begin{center}
\figspace

\includegraphics[width=11cm]{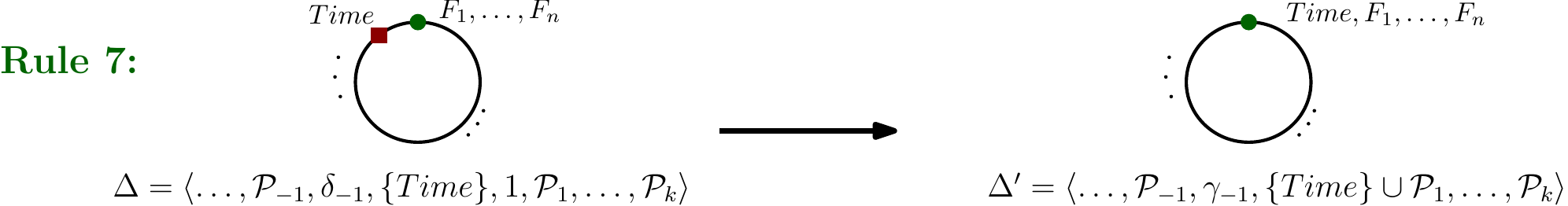}

%\textbullet\quad\textbf{\small Time alone and in the last class in unit circle and in $\Delta$ - Case 5:  $\gamma_{-1}$ is the truncated time of $\delta_{-1} + 1$:}
% \begin{center}
\figspace

% {\color{red}{
% $\Delta = \tup{\dots, \Pscr_{-1}, \delta_{-1}, \{Time\}} \qquad \longrightarrow \qquad \Delta = \tup{\dots, \Pscr_{-1}, \gamma_{-1}, \{Time\}}$
% }}
% \red{VN: Isn't this already captured by case 3, notice that k may be 0.}
%\includegraphics[width=11cm]{figures/rule-7.pdf}
% \end{center}
\vspace{-1mm}
 \caption{(Cont.) Rewrite Rules for Time Advancement using Circle-Configurations.
%  Here $\Delta$ is the $\delta$-configuration component of  the given Circle-Configuration.
 }
 \label{fig:time-advance-cont}
%\vspace{1em}
\end{figure}

\vspace{0,5em}
Time advancement rule $Tick$ is represented with a set of $Next$ rules, shown in Figure~\ref{fig:time-advance} and Figure~\ref{fig:time-advance-cont}. 
For a given circle-configuration, exactly one of the 8 $Next$ rules applies, depending on the position of 
the fact $Time$ on the unit circle $\Uscr$ with respect to the remaining facts.  For example, if the fact $Time$ is alone on the unit circle (and not at the zero point, nor in the last class), time advancement is modelled by placing $Time$ in the next class (clock-wise), see Rule~1. If we want to advance time from a  circle-configuration where $Time$ is in a class on a unit circle together with other facts (and not at the zero point, nor in the last class), we would place $Time$ alone on the unit circle, at any point  just before the next class (clock-wise) on the unit circle, see Rule~2. 
Cases when $Time$ is in the last class, in addition to changes in the unit circle, require updating of the $\delta$-configuration of the resulting circle-configuration, see  Figure~\ref{fig:time-advance-cont}. 

Since, application of a $Next$ rule changes the placement of the fact $Time$ on the unit circle w.r.t. remaining facts, %  of the given circle-configuration.  Consequently,
 the enabling and the resulting circle-configurations are different. Moreover, they represent configurations that may not not be equivalent. In fact resulting configuration is either equivalent to the enabling configuration or  is its immediate successor. % of the enabling configuration, 

%as stated in the proposition below.
Correspondence  to immediate successors refines our previous result~\cite[Lemma 1]{kanovich15post},
 stating that to a single $Tick$ rule corresponds a sequence of $Next$ rules, and, vice versa,  a sequence of $Next$ rules represents a single $Tick$ rule for an adequately chosen value $\varepsilon$ of  time advancement.
Here, we show how $Next$ rule relates to $Tick_{IS}$ rule.
% relates to  configuration corresponds exactly the application of a single $Next$ rule.

\begin{proposition}
\label{th: next-imm-succ}
Let $\Tscr $
%= \tup{\Sscr_0, \Rscr}$  
be an MSR with dense time,  $\GS$ a goal, $\CS$  a critical configuration specification and  $\Sscr_0$ an initial configuration.
Let $\Dmax$ be an upper bound on the numeric values  appearing in $\Tscr$,  $\GS$, $\Cscr\Sscr$ and $\Sscr_0$, and consider immediate successors of configurations w.r.t. the set of constraints from~$\Cscr_{\Dmax}$. 
\\ 
If \ \, $\Ascr_{1} \lra_{Next} {\Ascr_2}$  \ \, then \ \, %  $\Sscr_2$ is an immediate successor of  $\Sscr_1$  (\ie~ 
$\Sscr_1 \lra_{Tick_{IS}} \Sscr_2$, 
%~unless   ~$\{Time\}_\Zscr$, ~or~  \ $\{\}_\Zscr$ and ~$Time $~ is in the last class of the unit circle of $ {\Ascr_1}$,  in which cases  ~$\Sscr_1 \equiv \Sscr_2$.
\ or \ ~$\Sscr_1 \equiv \Sscr_2$ \ 
(in case $Next$ is Rule 0, for $n=0$, Figure~\ref{fig:time-advance}, \ or \ {Rule 4}, for $n=0$, Figure~\ref{fig:time-advance-cont}).
%~$\{Time\}_\Zscr$, ~or~   $\{\}_\Zscr$ and ~$Time $~ is in the last class of the unit circle of $ {\Ascr_1}$ .)
\\
If \ \, $\Sscr_1 \lra_{Tick_{IS}} \Sscr_2$ \ \, then \ \, $\Ascr_{1} \lra_{Next^n} {\Ascr_2}, \ n \in \{1,2,3\} \ . $
%If \ \, $\Sscr_1 \lra_{Tick_{IS}} \Sscr_2$ \ \, then \ \, $\Ascr_{1} \lra_{Next} {\Ascr_2}$, \ or \ $\Ascr_{1} \lra_{Next^n} {\Ascr_2}, \ n=2$ or $n=3$ \
%(in case $Next$ is Rule 0, for $n=0$, Figure~\ref{fig:time-advance}, \ or \ {Rule 4}, for $n=0$, Figure~\ref{fig:time-advance-cont}).
%~$\{Time\}_\Zscr$, ~or~   $\{\}_\Zscr$ and ~$Time $~ is in the last class of the unit circle of $ {\Ascr_1}$ .)
%(in case  ~$\{Time\}_\Zscr$, ~or~   $\{\}_\Zscr$ and ~$Time $~ is in the last class of the unit circle of $ {\Ascr_1}$.)
%Then, \ $\Ascr_{1} \lra_{Next} {\Ascr_2}$  \ \, if and only if \ \,  $\Sscr_2$ is an immediate successor of  $\Sscr_1$ w.r.t. the set of constraints from $\Cscr_{\Dmax}$ (\ie~ $\Sscr_1 \lra_{Tick_{IS}} \Sscr_2$). 
\end{proposition}
\begin{proof}
Both circle-configurations (\ie,~equivalence of configurations) and immediate successor configurations are defined w.r.t. an upper bound $\Dmax$. We set the value of $\Dmax$ to be an upper bound on numeric values in $\Tscr $, $\CS$  and  $\Sscr_0$ and refer to the same bound $\Dmax$ in both cases. Let     $\Ascr_1$  and $\Ascr_2$ be the circle-configurations of the configurations $\Sscr_1$ and $\Sscr_2$, respectively.

Notice that, as per Definition~\ref{def:circle-conf}, facts in the same class on the unit circle satisfy some constraint of the form $T_1=T_2\pm D$, while facts placed in different classes on the unit circle satisfy some constraint of the form $T_1< T_2\pm D$.

Let \,${\Ascr_1} \lra_{Next} {\Ascr_2}$. Then, as illustrated in Figure~\ref{fig:time-advance} and Figure~\ref{fig:time-advance-cont}, application of any of the 8 $Next$ rules, changes the placement of the fact $Time$ of the unit circle from one class to another. 

 There are two possibilities. In one case fact $Time$ is moved from a class containing some fact $F$ to a new class (see Rules 0,2,3). 
In the other case there exists some fact $F$ in $\Sscr_1$ such that $Time $ and $F$ are in different classes in $\Sscr_1$, but in the same class in  $\Sscr_2$ (see Rules1,4-7). 
%Clearly, in both cases, 
Configurations $\Sscr_1$ and $\Sscr_2$ do not satisfy the same constraints referring to facts $Time$ and $F$, except in the  two cases shown below:

\vspace{3pt} \noindent
\begin{minipage}[c]{0,48 \textwidth}
\includegraphics[width={\textwidth}]{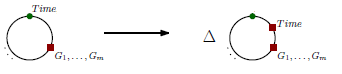}
\end{minipage}
\quad
\begin{minipage}[c]{0,48 \textwidth}
\includegraphics[width={\textwidth}]{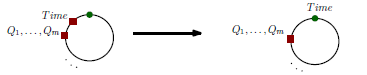}
\end{minipage}\\[5pt]
In the case shown to the left (Figure~\ref{fig:time-advance}, Rule 0, for $n=0$) fact $Time$ is the only fact placed at zero point, %\ie,~$\{Time\}_\Zscr$,
 while in the case shown to the right (Figure~\ref{fig:time-advance-cont}, \mbox{Rule 4}, for $n=0$) there are no facts at the zero point %\ie,~$\{\}_\Zscr$, 
 and ~$Time $~ is alone in the last class of the unit circle. % of $ {\Ascr_1}$.
Only in this two cases configurations are equivalent, ~$\Sscr_1 \equiv \Sscr_2$.

Moreover,  fact $Time$ is placed clock-wise, either to a position  immediately following its previous position, but before any existing class, or it is places exactly to the first class clock-wise. This ensures that there are no "intermediate" configurations, \ie~that $\Sscr_2$ is an immediate successor of  $\Sscr_1$, \ie,~ $\Sscr_1 \lra_{Tick_{IS}} \Sscr_2$, 
except in above two cases. % when ~$\Sscr_1 \equiv \Sscr_2$. 

\vspace{0,3em}
Conversely, if $\Sscr_2$ is an immediate successor of  $\Sscr_1$, then $\Sscr_1$ is transformed into $\Sscr_2$ by means of a $Tick_{IS}$ rule.
Then, when representing this time advancement with circle-configurations, the  placement of the facts different from $Time$ on the unit circle of  ${\Ascr_1}$  does not change. At the same time, the change in  placement of the fact $Time$ on the unit circle should be such to satisfy the condition of immediate successor configuration w.r.t the corresponding configurations. Figure~\ref{fig:time-advance} and Figure~\ref{fig:time-advance-cont} illustrate exactly such change in the placement of $Time$ on the unit circle, updating the $\delta$-configuration as well, when necessary.
The change  in placement of the fact $Time$ on the unit circle represents a minimal (or the exact) time advancement such that some constraint is no longer satisfied. Above two exceptions, related to the placement of the fact $Time$ at the zero point, require  2 or 3 $Next$ rules, as shown below:

\vspace{5pt} 
\includegraphics[width={0.95\textwidth}]{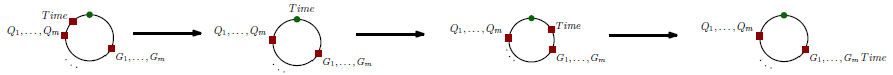}

\vspace{3pt} \noindent
Intermediate circle-configurations correspond to the configurations equivalent to the first one, but not to the final one.
\qed
\end{proof}

The above result ensures that the representation of time advancement on circle-configurations using~$Next$ rules is sound and complete. 
To $Next$ rules correspond $Tick_{IS}$ rules, and conversely, any $Tick_{\varepsilon}$ rule can be decomposed into a finite number of $Tick_{\varepsilon_i}$ rules (see Remark~\ref{remark-consecutive-ticks}),  each of which corresponds  to one, two or three $Next$ rules.
%To a single $Next$ rule corresponds a single $Tick$ rule for such a value $\varepsilon$ that a configuration is transformed into its immediate successor.
%Conversely, any $Tick_{\varepsilon}$ rule can be decomposed into a finite number of $Tick_{\varepsilon_i}$ rules (see Remark~\ref{remark-consecutive-ticks}),  each of which corresponds  to a single or two $Next$ rule.

 \vspace{0,5em}
We have considered traces over circle-configurations and showed that obtained traces  over circle-configurations are a sound and complete representation  of the set of traces over concrete configurations  with dense time.
Notice that circle-configurations are symbolic form, containing only untimed facts, a few auxiliary symbols and a bounded number of natural numbers. The are no real numbers included, and yet there is enough information for the sound and faithful representation of timed systems with dense time.
%, including traces over infinite time periods.
 This means that we can search for solutions of some problems symbolically, that is, 
without writing down the explicit values of the timestamps, \ie, the real numbers, in a trace.

 \vspace{0,5em}
In~\cite{kanovich15post} we investigated reachability problem which did not involve critical configurations.
%, nor infinite traces.
The notion of a non-critical trace in a timed MSR with dense time has not been investigated yet.  Since we now address the
non-critical reachability problem
% problems of realizability and survivability  
which involves non-critical traces, for our complexity results for timed MSR with dense time, we need to show that searching for traces in a symbolic form, using circle-configurations, is  sound and complete  also with respect to compliance \ie, preserves non-critical traces.

The notion of non-critical traces over  circle-configurations is not as complicated and delicate as the notion of a non-critical traces over configurations in systems with dense time, given in Definition~\ref{def:compliant}. 
Recall that the $Tick$ rule can be instantiated for any non-negative real value $\varepsilon$, denoting an arbitrary advancement of time, which can cause "skipping" over critical configurations. %, as shown in the above example.
Such a phenomena does not appear in traces over  circle-configurations where $Next$ rules are used for time advancement.
%
%
%This is because 
Following Proposition~\ref{th: IS critical}  and Proposition~\ref{th: next-imm-succ}, there is no %similar
 issue of "skipping" over critical circle-con\-fi\-gu\-ra\-tions with the time advancement $Next$. % as with the $Tick$ rule. 
%Given a circle-configuration only one of the 8 $Next$ rules given in Figures  \ref{fig:time-advance} and \ref{fig:time-advance-cont} is applicable. 
When a $Next$ rule is applied, the configuration corresponding to the resulting circle-configuration is an immediate successor of the  configuration corresponding to the  enabling configuration, or equivalent to it. That is, each of  8 $Next$ rules corresponds to a time advancement that is just enough, or exactly enough, so that some time constraint  involving the global time is no longer satisfied.
In such a way, a single $Next$ rule models either the minimal or the exact  advancement of time for which the equivalence class changes. 
Since there is no "skipping" over  circle-configurations, there is no need for decomposition of time advancements $Next$, as is the case with the $Tick$ rule. Hence, the related notion of compliance, \ie, non-critical traces, is straightforward.

\begin{definition}
\label{def: c-c compliance}
Let $\Tscr$ be a  a timed MSR with dense time and $\Cscr\Sscr$ a critical configuration specification.
% and $\Sscr_0$ an initial configuration.
  A trace over corresponding  {circle-configurations}    is \emph{non-critical} % for a given critical configuration specification
 if it does not contain any critical circle-configuration. 
\end{definition}

%\vspace{0,5em}
%The next proposition provides the bisimulation result on non-critical traces needed for non-critical reachability problem for traces over dense time domains.
%that use the lazy real-time sampling.

Recall that the notion of a non-critical trace in timed MSR with dense time %(Definition~\ref{def:compliant})
 potentially  involves  checking compliance through an infinite number of traces. Fortunately, this is not the case for non-critical traces over circle-configurations.
Since there is no "skipping" over  circle-configurations when using $Next$ rules, there is no need for decomposition of time advancements $Next$, as is the case with the $Tick$ rule. 
Smaller advancements of time would have either the exact same effect or no effect on a corresponding equivalence class.
On the other hand, larger advancements of time are modelled by a sequence of several $Next$ rules.
This is essential for the complexity of the problems involving non-critical traces, and we, therefore,  rely on non-critical traces over circle-configurations when searching for the solutions of our problems involving timed MSR with dense time.
The following proposition states that such a bisimulation is sound and complete w.r.t. application of rules and non-critical traces.
%, compliance and the lazy real-time sampling.

%Since, there is no need to check whether some critical circle-configurations were jumped over by the $Next$ rules
%
% With the application of a $Next$ rule, therefore, the configurations corresponding to the enabling and resulting circle-configurations are not equivalent.

%corresponds to a compliant plan over configurations.

%(Timed plan compliance problem) 
% Given a timed local state transition system $\Tscr$, an initial configuration $W$ consisting of timestamped facts and a finite, possibly empty, set of time constraints, a timed goal configuration $Z$, and a finite set of timed critical configurations, is there a compliant plan which leads from $W$ to $Z$?

% Formally, given a planning problem with the  \mbox{model $\Tscr$}, let $\Dmax$  be a  natural number such that $\Dmax > n + 1$ for any number $n$ appearing  in time constraints and actions in $\Tscr$,  and in the initial, goal and critical configurations. That is, $\Dmax$ is the upper bound on any numeric value  appearing either in a constraint ( $T_i \circ T_j \pm n$, for $\circ \in \{ <,>,= \leq, \geq \}$ ), in a timestamp of a fact in some action (\eg $F@(T+n)$ ), or in the timestamps of the initial configuration.

\begin{proposition}
\label{thm: circle configurations}
Given any timed MSR $\Tscr$ with dense time, a goal  $\GS$, a critical configuration specification $\Cscr\Sscr$  and  an initial configuration $\Sscr_0$, \
% the equivalence relation between configurations 
%given by Definition~\ref{def:equivalence}
% is well-defined with respect to the actions of the system (including time advances), lazy time scheduling and critical configurations. 
any non-critical trace starting from the given initial configuration  $\Sscr_0$ to a goal configuration can be conceived as a non-critical trace over  circle-configurations, starting from initial circle-configuration $\mathcal{A}_{\mathcal{S}_0}$  and reaching a goal circle-configuration. 
\end{proposition}

\begin{proof} 
In our previous work~\cite[Theorem 2]{kanovich15post} we have  shown a related bisimulation result for the reachability problem.
%There, traces were only finite, while now we consider possibly infinite traces.
%Additionally, for the  realizability and survivability problems 
Here we need to also address critical configurations. In particular, we must check time advancements more carefully in order to provide non-critical traces. 

To the given set of instantaneous rules of  timed MSR with dense time $\Tscr$, $r$, correspond the rules $[r]$  over circle-configurations, so that
% The $Tick$ rule is represented with the set of $Next$ rules. Let 
\ 
$$\widetilde{\Rscr} = \{ \ [r] : r \in \Rscr \ \} \cup Next $$ 
\ is the set of rules over circle-configurations.
Let $\Ascr_0$ be the circle-configuration of $\Sscr_0$.

%$\Fscr$  be the critical circle-configuration  corresponding to $\Cscr\Sscr$.

In~\cite[Theorem 2]{kanovich15post} we have shown that the equivalence among configurations is well defined with respect to application of rules. Namely, we have shown that for any instantaneous rule $r$, it is the case that
 $\Sscr_1 \lra_r \Sscr_2$ \  if and only  if \ \mbox{$\Ascr_1 \lra_{[r]} \Ascr_2$}, that is: 
\[
\begin{array}{cccccccccc}
\Sscr_1 & \to_{r}  &  \Sscr_{2} & % \qquad \qquad   & \ \Sscr'  & \to_{Tick}  & \Sscr''
\\
\arr & &   \arr &    %  \qquad \text{and} \qquad &  \arr &     &  \arr
\\
\Ascr_1 &\to_{[r]} & \Ascr_{2} & \ \ % . & % \qquad \qquad & \   \Cscr' & \to_{{Next}^*}  & \Cscr'' \ .
\end{array}
\]
where    $\Ascr_1$  and $\Ascr_2$ are circle-configurations of the configurations $\Sscr_1$, and $\Sscr_2$, respectively.
Also, it is the case that  \  $\Sscr_1 \lra_{Tick} \Sscr_2$ \  if and only if \ $\Ascr_1 \lra_{Next}^* \Ascr_2$, that is:
\[
\begin{array}{cccccccccc}
 \ \Sscr_1  & \to_{Tick}  & \Sscr_2
\\
  \arr &     &  \arr
\\
\   \Ascr_1 & \to_{{Next}^*}  & \Ascr_2 & \ \ .
\end{array}
\]
Again, $\Ascr_1$  and $\Ascr_2$ are circle-configurations of the configurations $\Sscr_1$, and $\Sscr_2$, respectively.
Notice that to each $Tick$ rule in the trace over configurations corresponds a (possibly empty) sequence of $Next$ rules in the matching trace over circle-configurations. 

Using induction on the length od a subtrace we can easily show that any trace of a timed MSR can be represented as a trace over corresponding circle-configurations, and vice versa, as shown below:
\[
%\begin{equation}
%\label{eq:bisimulation} 
\begin{array}{ccccccccccc}
\Pscr : & \ \ \Sscr_I = \Sscr_0 & \to_{r_1} \dots \to_{r_{i-1}} &  \Sscr_{i-1} & \to_{r_i} \  & \ \Sscr_{i}  & \to_{r_{i+1}}  \dots &  \to_{r_{l-1}} \  & \ \Sscr_{l}  
% \to_{r_{n}} & \Sscr_n=\Sscr_G
\\
& \arr & &   \arr &     &  \arr &  &   &  \arr
\\
\Pscr' : & \ \ \Ascr_I=\Ascr_0 &\to_{r_1'} \dots \to_{r_{i-1}'} & \Ascr_{i-1} & \to_{r_i'}&   \Ascr_i & \to_{r_{i+1}'} \dots &
\to_{r_{l-1}'} \  & \ \Ascr_{l}  
%\to_{r_{n}} &  \Ascr_n=\Ascr_G'
\end{array}
%\end{equation}
\]
where $r_i'$ is either the instantaneous rule $[r_i]$  over circle-configurations,  one or more $Next$ rules as given in Figures  \ref{fig:time-advance} and \ref{fig:time-advance-cont}, or an empty rule. 

\vspace{0,3em}
We can easily conclude that bisimulation preserves goals.
Since $ \Ascr_l$ is the circle-con\-fi\-gu\-ra\-ti\-on of  $\Sscr_l$, it immediately follows that $\Sscr_i $ 
 is a goal configuration  iff $\Sscr_i$ is a goal circle-configuration.
%  satisfy the same set of constraints  form $\Cscr_{\Dmax}$ as either $\Sscr$ or $\Sscr'$.
%This includes the constrains used in $\CS$.

\vspace{0,3em}
It remains to show that bisimulation preserves non-critical traces. For that purpose we decompose multiple $Next$ rules in $\Pscr'$.
As per   Proposition~\ref{th: next-imm-succ} the following correspondences for one or none applications of $Next$ rules holds:
\[
\begin{array}{ccccccccccccll}
 \ \Sscr  & \to_{Tick_{IS}}  &\  \Sscr    & &   & \ \Sscr  & \equiv  & \Sscr'   & &   & \ \Sscr  & \equiv  & \Sscr' 
\\
  \arr &     &  \arr &  \quad & \text{,} \quad &  \arr &     &  \arr    & \quad \text{or}  \quad & & \arr &     &  \arr
\\
\   \Ascr & \to_{Next}  & \Ascr' & \quad &  &  \   \Ascr & \to_{{Next}}  & \Ascr'   & \quad &  &  \   \Ascr & \to_{{Next}^0}  & \Ascr' = \Ascr
& \ \ .
\end{array}
\]
We can hence consider corresponding  traces $\Pscr $ and $\Pscr'$  as:
\[
%\begin{equation} \label{eq:bisimulation-1next}
\begin{array}{ccccccccccc}
\Pscr : & \ \ \Sscr_I =\bar{\Sscr}_0 & \to_{\bar{r}_1} \dots \to_{\bar{r}_{i-1}} &  \bar{\Sscr}_{i-1} & \to_{\bar{r}_i} \  & \ \bar{\Sscr}_{i}  & \to_{\bar{r}_{i+1}}  \dots &  \to_{\bar{r}_{n-1}} \  & \ \bar{\Sscr}_{n}
% \to_{r_{n}} & \Sscr_n=\Sscr_G
\\
& \arr & &   \arr &     &  \arr &  &   &  \arr
\\[3pt]
\Pscr' : & \ \ \Ascr_I=\bar{\Ascr}_0 &\to_{\bar{r}_1'} \dots \to_{\bar{r}_{i-1}'} & \bar{\Ascr}_{i-1} & \to_{\bar{r}_i'}&   \bar{\Ascr}_i & \to_{\bar{r}_{i+1}'} \dots &
\to_{\bar{r}_{n-1}'} \  & \ \bar{\Ascr}_{n}
%\to_{r_{n}} &  \Ascr_n=\Ascr_G'
\end{array}
%\end{equation}
\]
 where $\bar{r}_i'$ is either the instantaneous rule $[\bar{r}_i]$  over circle-configurations,  one $Next$ rule as given in Figures  \ref{fig:time-advance} and \ref{fig:time-advance-cont}, or an empty rule, and all $\bar{r}_i$  $Tick$ rules are either $Tick_{IS}$ rules or 
$Tick$ rules for which enabling and resulting configurations are equivalent.

% with all multiple $Next$ rules $[r_i']$  decomposed int
% so that rules into one or none applications of $Next$ rules. Then, 

\vspace{0,3em}
If the trace $\Pscr$ %starting  from $\Sscr_I$ %to $\Sscr_G$ 
is non-critical, all configurations $\bar{\Sscr}_i$ are not critical. 
Recall that a circle-configuration %$\Ascr_i $ %= \delta_{\Sscr_i}$
 is critical iff the corresponding configuration %$\Sscr_i$
 is critical.
Hence the corresponding circle-configurations $\bar{\Ascr}_i$ are not critical as well. Then the  above  trace $\Pscr'$
%starting  from $\Ascr_I$ 
 contains no critical circle-configuration and is therefore non-critical (Definition~\ref{def: c-c compliance}).

For the other direction, assume the trace $\Pscr'$  % starting from circle-con\-fi\-gu\-ra\-tion $\Cscr_I$
 is non-critical.
Then all circle-con\-fi\-gu\-ra\-ti\-ons $\bar{\Ascr}_i$ are not critical, and hence configurations $\bar{\Sscr}_i$ are not critical.
As per Definition~\ref{def:compliant}, we must consider decompositions of $Tick$ rules in $\Pscr$.

In the case ~$\bar{\Sscr}_{i}  \to_{Tick}  \bar{\Sscr}_{i+1}$~ and ~$\bar{\Sscr}_{i}  \equiv \bar{\Sscr}_{i+1}$,  following Remark~\ref{remark-consecutive-ticks},
such decompositions do not contain critical configurations since $\bar{\Sscr}_{i} $ and $\bar{\Sscr}_{i+1}$ are not critical.

In the other case, \ $\bar{\Sscr}_{i}  \to_{Tick_{IS}}  \bar{\Sscr}_{i+1}$. Then, from the Proposition~\ref{th: IS critical} we can conclude that
there are no critical configurations $\Sscr'$ such that  \ $\bar{\Sscr}_i \lra_{Tick} \Sscr' \lra_{Tick} \bar{\Sscr}_{i+1} \ $, \
$\forall i$.
%Therefore, as per Definition~\ref{def:compliant}, the trace $\Pscr$ is compliant.

%Using induction on the length of a subtrace we can now easily prove that any subtrace of $\Pscr,$  $\Sscr_I \lra_\Rscr^* \Sscr_k$, %containing only $Tick$ rules that reach immediate successor configurations, $Tick_{IS}$,  is compliant if and only if it contains no critical configuration.
%That is, the trace $\Pscr$ is non-critical (as per Definition~\ref{def:compliant})  if and only if the corresponding trace $\Pscr'$ over circle-configurations is non-critical.
\qed
\end{proof}

%\newpage

\subsection{PSPACE-Completeness of Non-critical Reachability Problem}
%Realizability and Survivability Problems in Dense Time Models }

\vspace{1em}
 PSPACE-hardness of  non-critical reachability problem
% both the realizability and survivability, 
%we adapt the proof of Proposition~\ref{th: realizability-hard}.
 can be  infered  from our 
%shown by adequately adapting our 
previous work~\cite{kanovich11jar} and \cite{kanovichJCS17} by considering non-critical reachability problem with no critical configurations.

\begin{proposition}
\label{th: non-critical-hard-real}
The   non-critical reachability problem timed MSR $\Ascr$ with dense time
% that uses the lazy real-time sampling are
is PSPACE-hard.
\end{proposition}

For   non-critical reachability problem we need to construct a non-critical trace from the given initial configuration to a goal configuration.
% that uses the lazy real-time sampling. 
As per Proposition~\ref{thm: circle configurations}, instead of non-critical traces over configurations of a given timed MSR with dense time, we can consider non-critical traces over circle-configurations.

The following lemma establishes a criteria related to the length of traces, that is 
% The following lemma establishes 
an upper bound
 on the number of different circle-configurations.

 \begin{lemma}\label{lemma:numstates-real}
Let $\Tscr$ be a timed MSR with dense time 
%  reachability problem $\Tscr$ under
constructed over  a finite alphabet $\Sigma$ with $J$ predicate symbols and $E$ constant and function symbols.
Let  $\CS$ a critical configuration specification, $\GS$ a goal, %configuration specification,
 $\Sscr_0$  be an initial configuration with $m$ facts, %   (counting  repetitions),
$\uSize$ an upper bound on the size of facts and $\Dmax$ an upper bound on the numeric values appearing in $\Tscr$, $\CS$, $\GS$ and 
$\Sscr_0$.

 Then the number of different circle-configurations, denoted
by $L(\iSize,\uSize,\Dmax)$, is
$$
 L(\iSize,\uSize,\Dmax) \leq J^\iSize  (E + 2 \iSize \uSize)^{\iSize
\uSize} \iSize^\iSize (\Dmax+2)^{(\iSize-1)} \ .
$$
% where $J$ and $D$ are, respectively, the number of predicate   and the number of constant and function symbols  in $\Sigma$.
\end{lemma}
\begin{proof} 
A circle-configuration consists of 
%two components, namely 
a $\delta$-configuration $\Delta$:
$$
 \Delta = \left\langle \
 \{Q_{1}^1, \ldots, Q_{m_1}^1\}, \,\delta_{1, 2},\, \{Q_{1}^2, \ldots, Q_{m_2}^2\}, \ldots ,  \,  \delta_{j-1, j}, \,\{Q_{1}^j, \ldots, Q_{m_j}^j\} \
 \right\rangle
$$
 and unit circle $\Uscr$:
$$
\Uscr = [\,\{Q_{1}^0, \ldots, Q_{m_0}^0\}_\Zscr,\, \{Q_{1}^1, \ldots, Q_{m_1}^1\}, \ldots, \,\{Q_{1}^j, \ldots, Q_{m_j}^j\}\,] \ .
$$

%  Let $\Delta$ be the $\delta$-configuration component of a circe-configuration. 
In each component, $\Delta$ and $\Uscr$,  there are $m$ facts, therefore there are  $\iSize$ slots for predicate names and at most $\iSize \uSize$ slots for  constants and function symbols. Constants can be either constants in the initial alphabet 
$\Sigma$ or names for fresh values (nonces). Following \cite{kanovich13ic} and Definition \ref{def:equivalence}, we need to consider only $2\iSize \uSize$ names for fresh values (nonces). 
Whenever an action creates some fresh values, instead of  new constants that have not yet appeared in the trace, we use nonce names from this fixed set,  different from any constants in the enabling configuration.  In that way, we are able to simulate an unbounded number of nonces using a set of only  $2mk$ nonce names.

For $\delta_{i, i+1}$, only the  time differences up to $\Dmax$ have to be considered together with the symbol $\infty$, and there
are at most $\iSize-1$ slots for time
differences $\delta_{i,j}$ in % $\delta$-representation 
$\Delta$. 

Finally, for each $\delta$-configuration, there are at most $\iSize^\iSize$ unit circles as for each 
fact $F$ we can assign a class, $\Uscr(F)$, and there are at most $\iSize$ classes. 
\qed
\end{proof}

% The proof follows from Lemma~\ref{lemma:numstates-real}.
Since, as per Lemma~\ref{lemma:numstates-real},  there are only $L(m,k,\Dmax)$ different circle-confi\-gu\-ra\-tions,  a non-critical trace $\Pscr$ of length greater than $L(m,k,\Dmax) $ necessarily contains the same circle-configuration $\Cscr$ twice, that is, there is a loop in the trace. 
Hence, there is a shorter plan that is the solution to the same non-critical reachability problem.
Therefore, we can nondeterministically search for plans of  length  bounded by $L(m,k,\Dmax)$.

%to show that the   non-critical reachability  problem is in PSPACE, we 
%%use Lemma~\ref{lem:lengthPSPACE-real} and
% search for non-critical  traces of length $L(m,k,\Dmax)$ (stored in binary). 
%%To do so, we rely on the fact that PSPACE and NPSPACE are the same complexity class~\cite{savitch}. 

\begin{theorem}
\label{th:PSPACE-feasibility-real}
 Assume $\Sigma$ a finite alphabet with $J$ predicate symbols and $E$ constant and function symbols, $\Tscr$ a MSR with dense time constructed over $\Sigma$, an initial configuration $\Sscr_0$  with $m$ facts, $\CS$ a critical configuration specification, $\GS$ a goal, %configuration specification,
 $k$ an upper-bound on the size of facts, and $\Dmax$ an upper-bound on the numeric values in $\Sscr_0, \Tscr$, $\CS$ and $\GS$.
%, and  functions $\next, \critical$ and \mustTickReal\ as described above. 
Let  functions $\next, \critical$ and  $\goal$ %\mustTickReal\  %as described above. 
  run in Turing space bounded by a polynomial in \ $m$, $k$, $\log_2(\Dmax)$ and return 1,  respectively, when a rule in $\Tscr$ is applicable to a given circle-configuration,   when a circle-configuration is critical with respect to $\CS$, and
when a circle-configuration is a goal circle-configuration with respect to $\GS$.
%$Tick$ rule should be applied to the given circle-configuration using the lazy time sampling.

There is an algorithm that, given an initial configuration $\Sscr_0$, decides whether non-critical trace
% w.r.t. to $\CS$
 in $\Tscr$ from $\Sscr_0$ to some goal configuration 
%$\Sscr_G$
 %w.r.t. $\GS$,
%is realizable with respect to the lazy real-time sampling, $\CS$ and $\Sscr_0$ 
and the algorithm runs in space bounded by a polynomial in $m,k$ and $log_2(\Dmax)$. 

The polynomial is in fact \ $\log_2(L(m,k,\Dmax))$.
\end{theorem}
\begin{proof}
We adapt the non-deterministic  algorithm used in~\cite[Teorem~7.3]{kanovichJCS17}
% the proof of Theorem~\ref{th:PSPACE-feasibility}
% given in \cite{kanovich13ic} and in \cite{kanovich15mscs} 
 in order to  obtain non-critical traces. %explicit dense time. 
The algorithm  accepts  whenever there is a non-critical trace which %uses the lazy real-time sampling,
 starts from $\Sscr_0$ and  reaches a goal configuration. 
%in which time tends to infinity.
We then apply Savitch's Theorem to determinize this algorithm. That is, we  rely
 on the fact that PSPACE and NPSPACE are the same complexity class~\cite{savitch}

Instead of searching for traces over concrete configurations, for the PSPACE result we rely on the equivalence among configurations and Proposition~\ref{thm: circle configurations}  which enable us to search for non-critical traces over circle-configurations, constructed using  the rules 
%$$%\widetilde{\Rscr} = \{ \ [r] : r \in \Tscr \ \} \cup Next \ .$$
$ [r]$, for $ r \in \Tscr $ and the $ Next $ rules.

% to return 1 when a rule $[r], r \in \Ascr$ is applicable to a given circle-configuration, when a $Next$  rule should be applied according to  the lazy real-time sampling,   and when a circle-configuration is critical with respect to $\CS$, respectively.

Because of Lemma~\ref{lemma:numstates-real}, it suffices to consider traces of size bounded by the number of different circle-configurations, ~$L(m,k,\Dmax) $ (stored in binary). Recall that
$$
 L(m,k,\Dmax) \leq J^\iSize  (E + 2 \iSize \uSize)^{\iSize \uSize} \iSize^\iSize (\Dmax+2)^{(\iSize-1)} \ .
$$

Let $i$ be a natural number such that \ $0 \leq i \leq L(m,k,\Dmax) +1$.
The algorithm starts with\ $i=0$ and  $\Ascr_0$ set as  the  circle-configuration of $\Sscr_0$, and iterates the following sequence of operations:

\begin{enumerate}
 \item If $\Ascr_i$ is a critical  circle-configuration, \ie, if $\critical(\Ascr_i) = 1$, then return FAIL, otherwise continue;
 \item If $\Ascr_i$ is a goal  circle-configuration, \ie, if $\goal(\Ascr_i) = 1$, then return ACCEPT, otherwise continue;
  \item If $i \geq L(m,k,\Dmax) $, then ACCEPT;
else continue;
 \item 
%If \mustTickReal $(\Ascr_i)=1$ then replace $\Ascr_i$ by $\Ascr_{i+1}$ obtained from $\Ascr_i$ by applying the $Next$ rule; Otherwise 
Non-deterministically guess an %instantaneous 
action, $r$, from $\Tscr$ applicable to $\Ascr_i$,
\ie, such an action $r$ that  $\next(r,\Ascr_i) = 1$. If  so
%no such action exists, then return FAIL. Otherwise
 replace $\Ascr_i$ with the  circle-configuration $\Ascr_{i+1}$
resulting from applying the action $[r]$ to the  circle-configuration
$\Ascr_i$. 
Otherwise FAIL;
%This is achieved by updating the facts, updating the positions of facts and the corresponding truncated time differences and continue;
  \item Set \ $ i = i + 1$.
\end{enumerate}

\vspace{3pt}
We now show that this algorithm runs in polynomial space.
The greatest number reached by the counter is $ L(m,k,\Dmax) $,  which stored in binary encoding takes 
space $ \log(L(m,k,\Dmax) + 1) $ \ bounded by:
\[
\begin{array}[]{lcl}
% \log(L(m,k,\Dmax) + 1) \leq  \\ \quad
m\log(J) +  mk\log(E + 2 mk) + m\log m  + (m-1)\log(\Dmax + 2).
\end{array}
\]
Therefore,  to store the values of the step-counter, one only needs space that is polynomial in the given inputs.

Also, any  circle-configuration, $\Ascr_i$ can be stored in space that is polynomial to the given inputs. 
Namely,  $\Ascr_i$ is of the form $ \tup{\Delta, \Uscr}$, with
$$
 \Delta = \left\langle \
 \{Q_{1}^1, \ldots, Q_{m_1}^1\}, \, \delta_{1, 2}, \, \{Q_{1}^2, \ldots, Q_{m_2}^2\}, \ldots ,  \,  \delta_{j-1, j},\, \{Q_{1}^j, \ldots, Q_{m_j}^j\} \
 \right\rangle
$$
 $$
\Uscr = [\,\{Q_{1}^0, \ldots, Q_{m_0}^0\}_\Zscr, \,\{Q_{1}^1, \ldots, Q_{m_1}^1\}, \ldots, \,\{Q_{1}^j, \ldots, Q_{m_j}^j\}\,] \ .
$$
Values of the truncated time differences, $\delta_{i,j}$, are bounded, so each $\delta$-configuration $\Delta$ can be stored in space $mk+(m-1)(\Dmax+2)$.
Each unit circle $\Uscr$ contains $m$ facts, % separated into classes,
 so it can be stored in  space $mk+(m-1)$, using a symbol for separating classes at most $(m-1)$ times.
Hence, each circle-configuration can be stored in space that is polynomially bounded with respect to the inputs.

Finally, in step 4.  algorithm needs to store the action $r$. This is done by remembering two circle-configurations. Moving from one  circle-configuration to another is achieved by updating the facts, updating the positions of facts and the corresponding truncated time differences. Hence, step 3. can be performed in space polynomial to  $m,k, log_2(\Dmax)$ and  the sizes of $\critical$, $\next$ and $\goal$.
Recall that functions $\critical$, $\next$ ~and  $\goal$ ~ run in space polynomial to the inputs.
\qed
\end{proof}

\begin{corollary}
  The   non-critical reachability problem for balanced timed MSR with dense time is PSPACE-complete when assuming a bound on the size of facts.
\end{corollary}

\emph{Acknowledgments:}
 Ban Kirigin is supported in part by the Croatian Science Foundation under the project UIP-05-2017-9219. Scedrov is partially supported by ONR. The participation of Kanovich and Scedrov in the preparation of this article was partially within the framework of the Basic Research Program at the National Research University Higher School of Economics (HSE) and partially supported within the framework of a subsidy by the Russian Academic Excellence Project `5-100'. Talcott is partly supported by ONR grant N00014-15-1-2202 and NRL grant N0017317-1-G002.

% \vspace{-5mm}

% \bibliographystyle{abbrv}
% \bibliography{master}

\end{document}